\documentclass{article}
\usepackage{amsmath}
\usepackage{amsthm} 
\usepackage{amssymb}
\usepackage{enumerate}
\usepackage{fullpage}
	
\usepackage{bussproofs}
\usepackage{url,xspace}
\usepackage{mathbbol}
\usepackage[all]{xy}
\SelectTips{cm}{}

\theoremstyle{plain} 
\newtheorem{theorem}{Theorem}[section]
\newtheorem{proposition}[theorem]{Proposition}
\newtheorem{lemma}[theorem]{Lemma}
\theoremstyle{definition}
\newtheorem{definition}[theorem]{Definition}
\newtheorem{example}[theorem]{Example}
\newtheorem{remark}[theorem]{Remark}

\newcommand{\To}{\Rightarrow}
\newcommand{\lto}{\leftarrow}
\newcommand{\pure}{{(0)}} 
\newcommand{\ppg}{{(1)}} 
\newcommand{\ctc}{{(2)}} 
\newcommand{\dec}{{(d)}} 
\newcommand{\eqs}{\equiv} 
\newcommand{\eqw}{\sim} 
\newcommand{\empt}{\mathbb{0}}
\newcommand{\inn}{\mathit{normal}} 
\newcommand{\ina}{\mathit{abrupt}} 
\newcommand{\copi}{q} 
\newcommand{\cotuple}[1]{\left[ #1 \right]}
\newcommand{\bigcotuple}[1]{\bigl[ #1 \bigr]}
\newcommand{\Bigcotuple}[1]{\Bigl[ #1 \Bigr]}
\newcommand{\cotu}{\cotuple{\,}}
\newcommand{\toppg}[1]{\blacktriangledown #1 } 
\newcommand{\bigtoppg}[1]{\blacktriangledown #1 } 
\newcommand{\throw}[2]{\mathit{throw}_{#2,#1}} 
\newcommand{\try}[2]{\mathit{try}\{#1\}\,#2}  
\newcommand{\TRY}[2]{\mathit{TRY}\{#1\}\,#2}  
\newcommand{\catch}[2]{\mathit{catch}\,\{#1 \To #2\}}  
\newcommand{\catchn}[4]{\mathit{catch}\,\{#1\!\To\! #2|\dots|#3\!\To\! #4\}}  
\newcommand{\catchij}[4]{\mathit{catch}\,\{#1\!\To\! #2\;|\;#3\!\To\! #4\}}  
\newcommand{\Ss}{\Sigma} 
\newcommand{\Tt}{\Theta} 
\newcommand{\Hh}{\mathcal{H}} 
\newcommand{\Cc}{\mathcal{C}} 
\renewcommand{\ss}{\sigma} 
 
\newcommand{\pu}{\mathit{pure}} 	
\newcommand{\exc}{\mathit{exc}} 
\newcommand{\meq}{\mathit{meq}} 	
\newcommand{\app}{\mathit{app}} 
\newcommand{\deco}{\mathit{deco}} 
\newcommand{\expl}{\mathit{expl}} 
\newcommand{\undeco}{d} 
\newcommand{\expand}{e} 
\newcommand{\Set}{\mathbf{Set}}
\newcommand{\catC}{\mathbf{C}}
\newcommand{\catS}{\mathbf{S}} 
\newcommand{\catT}{\mathbf{T}} 
\newcommand{\eps}{\varepsilon}
\newcommand{\id}{\mathit{id}}
\newcommand{\Id}{\mathit{Id}}
\newcommand{\skE}{\mathbf{E}}
\newcommand{\ske}{\mathbf{e}}
\newcommand{\Rea}{\mathit{Real}}
\newcommand{\DiaL}{\mathcal{L}}
\newcommand{\DiaR}{\mathcal{R}}
\newcommand{\DiaF}{\mathcal{F}}
\newcommand{\DiaG}{\mathcal{G}}
\newcommand{\Yon}{\mathcal{Y}}
\newcommand{\Sig}{\mathit{Sig}} 	
\newcommand{\Mod}{\mathit{Mod}}

\newcommand{\fra}[2]{#2 \backslash #1} 
\newcommand{\unit}{\mathbb{1}}

\title{Adjunctions for exceptions} 
\author{
Jean-Guillaume Dumas
\thanks{This work is partly funded by the HPAC project 
of the French Agence Nationale de la Recherche (ANR 11 BS02 013).}, 
  Dominique Duval
\thanks{This work is partly funded by the CLIMT project 
of the French Agence Nationale de la Recherche (ANR 11 BS02 016).}, 
Laurent Fousse, Jean-Claude~Reynaud
\\ Universit\'e de Grenoble, Laboratoire Jean Kuntzmann
\\ \{jgdumas,dduval,lfousse,jcreynaud\}@imag.fr}
\date{October 26, 2012}

\begin{document}

\maketitle

\begin{abstract} 
\textbf{Abstract.}
The exceptions form a computational effect,
in the sense that there is an apparent mismatch between 
the syntax of exceptions and their intended semantics.
We solve this apparent contradiction by 
defining a logic for exceptions 
with a proof system which is close to their syntax
and where their intended semantics can be seen as a model. 
This requires a robust framework for logics and their morphisms, 
which is provided by diagrammatic logics. 

Semantics of exceptions. 
Adjunction.
Categorical fractions.
Limit sketches.
Diagrammatic logics. 
Morphisms of logics. 
Decorated proof system.
\end{abstract}

\section*{Introduction}

Exceptions form a \emph{computational effect}, 
in the sense that a syntactic expression $f:X\to Y$ 
is not always interpreted as a function $f:X \to Y$: 
for instance a function which raises an exception 
has to be interpreted as a function $f:X \to Y+E$ 
where $E$ is the set of exceptions. 
In a computer language usually 
exceptions differ from errors in the sense that it is possible to recover from
an exception while this is impossible for an error;
thus, exceptions have to be both raised and handled. 
Besides, the theory of \emph{diagrammatic logics} forms a new paradigm 
for understanding the nature of computational effects;
in this paper, diagrammatic logics are applied to the 
denotational semantics of exceptions.

To our knowledge, the first categorical treatment of 
computational effects is due to Moggi \cite{Mo91};
this approach relies on \emph{monads}, it is implemented 
in the programming language Haskell \cite{Wa92,haskell}. 
The examples proposed by Moggi include 
the states monad $T X = (X\times S)^S$ where $S$ is the set of states and 
the exceptions monad $T X = X+E$ where $E$ is the set of exceptions. 
Later on, using the correspondence between monads
and algebraic theories, Plotkin and Power proposed to 
use \emph{Lawvere theories} for dealing with 
the operations and equations related to computational effects, 
for instance the lookup and update operations for states and 
the raising and handling operations for exceptions \cite{PP02,HP07}.
In the framework of Lawvere theories,   
an operation is called \emph{algebraic} 
when it satisfies some relevant genericity properties; 
the operations lookup and update for states 
and the operation for raising exceptions are algebraic in this sense, 
while the operation for handling exceptions is not \cite{PP03}.
This difficulty can be overcome,  
as for instance in \cite{SM04,Le06,PP09}.
Nevertheless, from these points of view the handling of exceptions is
inherently different from the updating of states.

In this paper we use another method for dealing with 
computational effects. 
This method has been applied to a parametrization process 
in \cite{DD10-dialog,DD12-param} 
and to the states effect in \cite{DDFR12b-state}.
It has led to the discovery of a duality 
between states and exceptions, briefly presented in \cite{DDFR12a-short}. 
Our approach also provides a notion of \emph{sequential product}, 
which is an alternative to the strength of a monad
for imposing an evaluation order for the arguments of a $n$-ary function
\cite{DDR11-seqprod}. 
With this point of view the fact that 
the handling operation for exceptions is not algebraic, 
in the sense of Lawvere theories, is not an issue. 
In fact, the duality between the exceptions effect and the states effect
\cite{DDFR12a-short} implies 
that catching an exception is dual to updating a state. 
It should be noted that we distinguish 
the private operation of catching an exception from 
the public operation of handling it (also called ``try/catch''),   
which encapsulates the catching operation.

Our idea is to look at an effect as an apparent mismatch between 
syntax and semantics: 
there is one logic which fits with the syntax, 
another one which fits with the semantics, 
and a third one which reconciles syntax and semantics. 
This third logic classifies the language features and their properties 
according to the way they interact with the effect; 
we call this kind of classification a \emph{decoration}. 
For this conciliation, as the features of the different logics are
quite different in nature, we will use morphisms from the decorated
logic to each of the two other logics, in a relevant category. 

This approach thus requires a robust framework 
for dealing with logics and morphisms between them. 
This is provided by the category of \emph{diagrammatic logics} 
\cite{D03-diaspec,DD10-dialog}. 
The main ingredient for defining this category is the notion
of  categorical \emph{fraction}, as introduced in  \cite{GZ67}  
for dealing with homotopy theory. Fractions are defined 
with respect to an \emph{adjunction}.
The syntactic aspect of logics is obtained by assuming that this 
adjunction is induced by a morphism of \emph{limit sketches} \cite{Eh68}, 
which implies that the adjunction connects \emph{locally presentable} 
categories. 
For each diagrammatic logic we define \emph{models} as relevant morphisms,   
\emph{inference rules} as fractions and \emph{inference steps}
as composition of fractions. 
Thus, diagrammatic logics are defined from well-known 
categorical ingredients; their novelty lies in the importance 
given to fractions, in the categorical sense, for formalizing logics. 

The category of diagrammatic logics is introduced in Section~\ref{sec:dialog}.
In Section~\ref{sec:expl} we look at exceptions 
from an \emph{explicit} point of view, 
by introducing a type of exceptions in the return type of operations 
which may raise exceptions. 
With this explicit point of view we formalize 
(by Definition~\ref{defi:expl-model}) 
the intended semantics of exceptions  
as provided in the documentation of the computer language Java \cite{java}
and reminded in Appendix~\ref{app:java}. 
We also introduce the distinction between the core operations 
and their encapsulation:  
typically, between the catching and the handling of exceptions. 
This explicit point of view is expressed in terms of 
a diagrammatic logic denoted $\DiaL_\expl$: 
the intended semantics of exceptions  
can be seen as a model with respect to $\DiaL_\expl$. 
Then in Section~\ref{sec:deco} we look at exceptions from a 
\emph{decorated} point of view, which fits with the syntax 
much better than the explicit point of view, since  
the return type of operations does not mention any type of exceptions. 
The key point in this logic is that the operations and equations are 
decorated according to their interaction with exceptions. 
This decorated point of view corresponds to another 
diagrammatic logic denoted $\DiaL_\deco$. 
We build a morphism of diagrammatic logics 
from $\DiaL_\deco$ to $\DiaL_\expl$, called the \emph{expansion}, 
from which our main result (Theorem~\ref{thm:deco-model}) follows: 
the intended semantics of exceptions  
can also be seen as a model with respect to $\DiaL_\deco$. 
  $$ \xymatrix@C=6pc@R=3pc{
  \DiaL_\deco \ar[r]^{\txt{expansion}} 
     \ar@{~>}[d]_(.4){\txt{model}}_(.7){\txt{(Section~\ref{sec:deco})}} & 
  \DiaL_\expl 
     \ar@{~>}[d]^(.4){\txt{model}}^(.7){\txt{(Section~\ref{sec:expl})}} \\
  \txt{semantics} \ar@{=}[r]_{\txt{(Theorem~\ref{thm:deco-model})}} & 
     \txt{semantics} \\ 
  }  $$ 
Then we prove some properties of exceptions using 
the rules of the decorated logic 
and the duality between exceptions and states. 
We conclude in Section~\ref{sec:conc} with some remarks
and guidelines for future work. 

\section{The category of diagrammatic logics}
\label{sec:dialog}

This paper relies on the robust algebraic framework 
provided by the category of diagrammatic logics 
\cite{DD10-dialog,D03-diaspec}. 
In Section~\ref{subsec:dialog-adjoint} we provide 
an informal description of diagrammatic logics which 
should be sufficient for understanding most of Sections~\ref{sec:expl} 
and~\ref{sec:deco}.
Let us also mention the paper \cite{DD12-param} 
for a detailed presentation of a simple application of diagrammatic logics. 
Precise definitions of diagrammatic logics and their morphisms
are given in Section~\ref{subsec:dialog-dialog}; 
these definitions rely on the categorical notions
of fractions and limit sketches.

\subsection{A diagrammatic logic is a left adjoint functor}
\label{subsec:dialog-adjoint}

In this Section we give 
an informal description of diagrammatic logics and their morphisms; 
the formal definitions will be given in Section~\ref{subsec:dialog-dialog}.
In order to define a diagrammatic logic we have to distinguish 
its \emph{theories}, which are closed under deduction, 
from its \emph{specifications}, which are presentations of theories. 
On the one hand, each specification generates a theory, 
by applying the \emph{inference rules} of the logic:
the specification is a family of axioms
and the theory is the family of theorems which can be proved
from these axioms, using the inference system of the logic.  
On the other hand, each theory can be seen as a (``large'') specification.

Then, clearly, a \emph{morphism} of logics has to map 
specifications to specifications
and theories to theories, in some consistent way.  
The diagrammatic logics we are considering in this paper 
are variants of the equational logic:   
their specifications are made of (some kinds of) 
sorts, operations and equations. 
Each sort, operation or equation can be seen as a specification,
hence every morphism of logics has to map it to a specification.
However it is not required that a sort be mapped to a sort, 
an operation to an operation or an equation to an equation. 
Thanks to this property, rather subtle relations between logics
can be formalized by morphisms of 
diagrammatic logics. This is the case for the expansion morphism, 
see Figure~\ref{fig:expansion}.

A diagrammatic logic is a left adjoint functor $\DiaL$
from a category $\catS$ of specifications
to a category $\catT$ of theories,
with additional properties that will be given 
in Section~\ref{subsec:dialog-dialog}.
Each specification generates a theory thanks to this functor $\DiaL$ 
and each theory can be seen as a specification thanks to 
the right adjoint functor $\DiaR:\catT\to\catS$. 
In addition, it is assumed that the canonical morphism 
$\eps_\Tt: \DiaL\DiaR\Tt \to \Tt$ is an isomorphism in $\catT$, 
so that each theory $\Tt$ can be seen as a presentation of itself. 
The fact that indeed the functor $\DiaL$ describes an \emph{inference system} 
is due to additional assumptions on the adjunction 
$\DiaL\dashv\DiaR$, which are given in the next Section.  

Although this point will not be formalized, 
in order to understand the definition of the models of a specification 
it may be helpful to consider that one is 
usually interested in two kinds of theories: 
the theories $\DiaL\Ss$ which are generated by a ``small'' (often finite) 
specification $\Ss$,  
and the ``large'' theories $\Tt$ which are provided by set theory, 
domain theory and other mathematical means, to be used as 
interpretation domains. 
However, formally this distinction is useless, and 
the \emph{models} of any specification $\Ss$ with values in any theory $\Tt$ 
are defined as the morphisms from $\DiaL\Ss$ to $\Tt$ in $\catT$. 
Thanks to the adjunction $\DiaL\dashv\DiaR$, there is an alternative 
definition which has a more constructive flavour:
the \emph{models} of $\Ss$ with values in $\Tt$ 
are the morphisms from $\Ss$ to $\DiaR\Tt$ in $\catS$.

The definition of \emph{morphisms} between diagrammatic logics 
derives in an obvious way from the definition of diagrammatic logics:
a morphism $F:\DiaL_1\to\DiaL_2$ is made of two functors 
$F_S:\catS_1\to\catS_2$ and $F_T:\catT_1\to\catT_2$ 
with relevant properties. 

In this paper we consider several diagrammatic logics 
which are variants of the equational logic: 
the specifications are defined from sorts, operations and equations,
and the inference rules are variants of the usual equational rules.  
Exceptions form a \emph{computational effect}, 
in the sense that a syntactic expression $f:X\to Y$ 
may be interpreted as a function $f:X \to Y+E$ 
(where $E$ is the set of exceptions) 
instead of being interpreted as a function $f:X \to Y$. 
We will define 
a logic $\DiaL_\deco$ for dealing with the syntactic expressions  
and another logic $\DiaL_\expl$ for dealing with exceptions in an explicit way  
by adding a sort of exceptions (also denoted $E$).
The key feature of this paper is the \emph{expansion} morphism
form the logic $\DiaL_\deco$ to the logic $\DiaL_\expl$, 
which maps a syntactic expression $f:X\to Y$ 
to the expression $f:X\to Y+E$ in a consistent way. 

\subsection{Diagrammatic logics, categorically}
\label{subsec:dialog-dialog}

The notion of diagrammatic logic is an algebraic notion 
which captures some major properties of logics 
and which provides a simple and powerful notion 
of morphism between logics. 
Each diagrammatic logic comes with a notion of models 
and it has a sound inference system. 

A category is \emph{locally presentable} when it is equivalent to 
the category $\Rea(\skE)$ of set-valued models, or \emph{realizations}, 
of a limit sketch $\skE$ \cite{Eh68,GU71}. 
The category $\Rea(\skE)$ has colimits and 
there is a canonical contravariant functor $\Yon$ from $\skE$ to $\Rea(\skE)$, 
called the \emph{contravariant Yoneda functor} associated with $\skE$, 
such that $\Yon(\skE)$
generates $\Rea(\skE)$ under colimits, in the sense that 
every object of $\Rea(\skE)$ may be written as a colimit 
over a diagram with objects in $\Yon(\skE)$.  

Each morphism of limit sketches $\ske:\skE\to\skE'$ 
gives rise, by precomposition with $\ske$, to a functor 
$G_\ske:\Rea(\skE')\to\Rea(\skE)$,
which has a left adjoint $F_\ske$ \cite{Eh68}. 
Let $\Yon$ and $\Yon'$ denote the contravariant Yoneda functors 
associated with $\skE$ and $\skE'$, respectively.  
Then $F_\ske$ extends $\ske$, which means that 
$F_\ske\circ\Yon=\Yon'\circ\ske$ up to a natural isomorphism. 
We call such a functor $F_\ske$ a \emph{locally presentable functor}. 
Then the three following properties are equivalent: 
the counit $\eps:F_\ske\circ G_\ske\To\Id$ is a natural isomorphism;  
the right adjoint $G_\ske$ is full and faithful;
the left adjoint $F_\ske$ is (up to an equivalence of categories) 
a \emph{localization}, 
i.e., it consists of adding inverses to some morphisms from $\Rea(\skE)$,
constraining them to become isomorphisms in $\Rea(\skE')$ \cite{GZ67}. 
Then it can be assumed that $\ske$ is also a localization: 
it consists of adding inverses to some morphisms from $\skE$. 

\begin{definition}
\label{defi:dialog-dialog}
A \emph{diagrammatic logic} is a locally presentable functor 
which is a localization, up to an equivalence of categories. 
\end{definition}

It follows that a diagrammatic logic is a left adjoint functor 
such that its counit is a natural isomorphism: 
these properties have been used in Section~\ref{subsec:dialog-adjoint}.

\begin{definition}
\label{defi:dialog-logic-model}
Let $\DiaL:\catS\to\catT$ be a diagrammatic logic with right adjoint $\DiaR$.
\begin{itemize}
\item The category of \emph{$\DiaL$-specifications} is $\catS$. 
\item The category of \emph{$\DiaL$-theories} is $\catT$. 
\item A \emph{model} of a specification $\Ss$ 
with values in a theory $\Tt$ is 
a morphism from $\DiaL\Ss$ to $\Tt$ in~$\catT$,
or equivalently (thanks to the adjunction) 
a morphism from $\Ss$ to $\DiaR\Tt$ in $\catS$. 
\end{itemize}
\end{definition}

The \emph{bicategory of fractions} associated to $\DiaL$ 
has the same objects as $\catS$ and 
a morphism from $\Ss_1$ to $\Ss_2$ in this bicategory 
is a \emph{fraction} $\fra{\ss}{\tau}:\Ss_1\to\Ss_2$,  
which means that it is a cospan 
$(\ss:\Ss_1\to\Ss'_2\lto\Ss_2:\tau)$ in $\catS$ 
such that $\DiaL\tau$ is invertible in $\catT$. 
Then $\ss$ is called the \emph{numerator} and $\tau$ the \emph{denominator} 
of the fraction $\fra{\ss}{\tau}$. 
It follows that we can define 
$\DiaL(\fra{\ss}{\tau})=\DiaL\tau^{-1}\circ\DiaL\ss$. 
The composition of consecutive fractions is defined 
as the composition of cospans, using a pushout in $\catS$. 

\begin{definition}
\label{defi:dialog-logic-rule}
Let $\DiaL:\catS\to\catT$ be a diagrammatic logic with right adjoint $\DiaR$.
\begin{itemize}
\item A \emph{rule} with \emph{hypothesis} $\Hh$ 
and \emph{conclusion} $\Cc$ is a fraction from $\Cc$ to $\Hh$
with respect to $\DiaL$. 
\item An \emph{instance} of a specification $\Ss_0$ in a specification $\Ss$ 
is a fraction from $\Ss_0$ to $\Ss$ 
with respect to $\DiaL$. 
\item The \emph{inference step}
applying a rule $\rho:\Cc\to\Hh$ to an instance 
$\iota:\Hh\to\Ss$ of $\Hh$ in $\Ss$ is the composition 
of fractions $\iota \circ \rho:\Cc\to\Ss$; 
it yields an instance of $\Cc$ in $\Ss$. 
\end{itemize}
\end{definition}

\begin{definition}
\label{defi:dialog-logic-proof}
Let $\DiaL:\catS\to\catT$ be a diagrammatic logic with right adjoint $\DiaR$.
\begin{itemize}
\item Each morphism of limit sketches $\ske:\skE_S\to\skE_T$ 
which gives rise to the adjunction $\DiaL\dashv\DiaR$ 
and which is a localization is called an \emph{inference system} for $\DiaL$. 
\item Then a rule $\fra{\ss}{\tau}$ is \emph{elementary} 
if $\ss$ and $\tau$ are the images,
by the canonical contravariant functor $\Yon$,  
of arrows $s$ and $t$ in $\skE_S$ 
such that $\ske(t)$ is invertible in $\skE_T$;
otherwise the rule $\fra{\ss}{\tau}$ is \emph{derivable}. 
\end{itemize}
\end{definition}

\begin{remark}
\label{rem:fractions}
An inference rule is usually written as a fraction 
$\frac{\Hh_1\dots\Hh_k}{\Cc}$, 
it is indeed related to a categorical fraction, 
as follows (however from the categorical point of view 
the numerator is on the conclusion side
and the denominator on the hypothesis side!). 
First let us remark that each $\Hh_i$ can be seen as a specification,
as well as $\Cc$, and that the common parts 
in the $\Hh_i$'s and in $\Cc$ are indicated by using the same names. 
Then let $\Hh$ be the vertex of the colimit of the $\Hh_i$'s, 
amalgamated according to their common names. 
The fraction $(\ss:\Cc\to\Hh'\lto\Hh:\tau)$ 
is defined as the pushout of $\Hh$ and $\Cc$ 
over their common names. 
Then the rule $\frac{\Hh_1\dots\Hh_k}{\Cc}$ 
corresponds to the categorical fraction $\fra{\ss}{\tau}:\Cc\to\Hh$
(see Example~\ref{exam:dialog-meq}).
In an inference system $\ske:\skE_S\to\skE_T$ for a logic $\DiaL$, 
the limit sketch $\skE_S$ describes the syntax and the morphism $\ske$ 
provides the inference rules of $\DiaL$. 
Thus, the description of a diagrammatic logic 
via one of its inference systems can be done 
algebraically by defining $\ske$ or the image of $\ske$ 
by the canonical funtor $\Yon$ 
(examples can be found in \cite{DD12-param}).
A diagrammatic logic can also be defined 
more traditionally by giving a grammar and a family of rules.
Moreover, when the logic is simple enough, 
it may be sufficient in practice to describe its theories.
\end{remark}

\begin{example}[Monadic equational logic]
\label{exam:dialog-meq}
The monadic equational logic $\DiaL_\meq$ can be defined from its theories. 
A monadic equational theory 
is a category where the axioms hold only up to some congruence relation.
Precisely, a \emph{monadic equational theory} is a directed graph 
(its vertices are called \emph{types} and its edges are called \emph{terms})
with an \emph{identity} term $\id_X:X\to X$ for each type $X$ 
and a \emph{composed} term $g\circ f:X\to Z$ for each pair 
of consecutive terms $(f:X\to Y,g:Y\to Z)$;
in addition it is endowed with \emph{equations} $f\eqs g:X\to Y$
that form an equivalence relation on parallel terms 
which is a \emph{congruence} with respect to the composition 
and such that the associativity and identity axioms hold up to 
congruence. 
The category of sets forms a $\DiaL_\meq$-theory $\Set$ 
where types, terms and equations are the sets, functions and equalities. 
We can look at a rule, for instance the transitivity rule for equations 
$\frac{f \eqs g \quad g \eqs h}{f \eqs h}$, 
as a categorical fraction  $\fra{\ss}{\tau}:\Cc\to\Hh$, as follows.  
$$ \begin{array}{|c|c|c|c|c|}
\cline{1-1} \cline{3-3} \cline{5-5} 
\Cc && \Hh' && \Hh \\
\cline{1-1} \cline{3-3} \cline{5-5} 
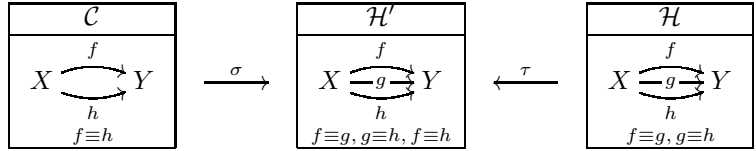
\xymatrix@R=.2pc{
X \ar@/^/[r]^{f} \ar@/_/[r]_{h} & Y \\ \\ 
\ar@{}[r]|{f\equiv h} & \\ 
}  &
\xymatrix@=2pc{\ar[r]^{\ss} & \\ }  &
\xymatrix@R=.2pc{
X \ar@/^/[r]^{f} \ar[r]|{\,g\,} \ar@/_/[r]_{h} & Y \\ \\ 
\ar@{}[r]|{f\equiv g,\, g\equiv h,\, f\equiv h} & \\ 
}   &
\xymatrix@=2pc{ & \ar[l]_{\tau} \\ }  &
\xymatrix@R=.2pc{
X \ar@/^/[r]^{f} \ar[r]|{\,g\,} \ar@/_/[r]_{h} & Y \\ \\ 
\ar@{}[r]|{f\equiv g,\, g\equiv h} & \\ 
}    \\
\cline{1-1} \cline{3-3} \cline{5-5} 
\end{array} $$
\end{example} 

\begin{remark}
\label{rem:doctrines}
Diagrammatic logics generalize \emph{$E$-doctrines}, 
in the sense of \cite{We94}.
Let $E$ be a type of sketch, determined by what sorts 
of cones and cocones are allowed in the sketch. 
Then $E$ determines a type of category, 
required to have all (co)limits of the sorts of (co)cones allowed by $E$,
and it  determines a type of functor, required to preserve that 
sorts of (co)limits. 
Following \cite{We94}, 
the \emph{$E$-doctrine} is made of these sketches, categories and functors.
Each $E$-doctrine corresponds to a diagrammatic logic 
$\DiaL_E:\catS_E\to\catT_E$, 
where $\catS_E$ is the category of $E$-sketches (with the morphisms 
of $E$-sketches), 
$\catT_E$ is the category of $E$-categories and $E$-functors,
and $\DiaL_E$ is the left adjoint functor which maps each 
$E$-sketch to its \emph{theory}. 
For instance the $E$-doctrine made of finite products sketches, 
cartesian categories and functors preserving finite products
corresponds to the equational logic. 
\end{remark}

An important feature of diagrammatic logics is their simple and powerful
notion of morphism, which is a variation of the notion of morphism 
in an arrow category. 

\begin{definition}
\label{defi:dialog-morphism}
Given diagrammatic logics $\DiaL:\catS\to\catT$ and $\DiaL':\catS'\to\catT'$, 
a \emph{morphism of diagrammatic logics} 
$\DiaF:\DiaL\to \DiaL'$ is made of two locally presentable functors 
$\DiaF_S:\catS\to\catS'$ and $\DiaF_T:\catT\to\catT'$ 
such that the square of left adjoints $(\DiaL,\DiaL',\DiaF_S,\DiaF_T)$ 
is induced by a commutative square of limit sketches. 
It follows that the right adjoints form a commutative square 
and that the left adjoints form a square which is commutative up to a 
natural isomorphim. 
\end{definition}

This means that a morphism from $\DiaL$ to $\DiaL'$ 
maps (in a coherent way) 
each specification of $\DiaL$ to a specification of $\DiaL'$ 
and each proof of $\DiaL$ to a proof of $\DiaL'$. 
Moreover, it is sufficient to check 
that each elementary specification (i.e., 
each specification in the image of the functor $\Yon$)
of $\DiaL$ is mapped to a specification of $\DiaL'$
and that each elementary proof (i.e., each inference rule) 
of $\DiaL$ is mapped to a proof of $\DiaL'$. 
The next result is the key point for proving Theorem~\ref{thm:deco-model}; 
its proof is a straightforward application of the properties of adjunctions. 

\begin{proposition}
\label{prop:deco-morphism} 
Let $\DiaF=(\DiaF_S,\DiaF_T):\DiaL\to \DiaL'$ be 
a morphism of diagrammatic logics 
and let $\DiaG_T$ be the right adjoint of $\DiaF_T$. 
Let $\Ss$ be a $\DiaL$-specification and $\Tt'$ a $\DiaL'$-theory. 
Then there is a bijection, natural in $\Ss$ and $\Tt'$:
  $$  \Mod_{\DiaL}(\Ss,\DiaG_T\Tt') \cong \Mod_{\DiaL'}(\DiaF_S\Ss,\Tt') \;.$$
\end{proposition}

\section{Denotational semantics of exceptions}
\label{sec:expl}

In this Section we define a denotational semantics 
of exceptions which relies on the semantics of exceptions in Java. 
Syntax is introduced in Section~\ref{subsec:expl-sig} 
as a signature $\Sig_\exc$.
The fundamental distinction between ordinary and exceptional values is 
discussed in Section~\ref{subsec:expl-sem}.  
Sections~\ref{subsec:expl-expl} and~\ref{subsec:expl-spec} 
are devoted to the definitions of  
a logic with an explicit type of exceptions 
and a specification $\Ss_\expl$ for exceptions with respect to this logic. 
Then in Section~\ref{subsec:expl-model}  
the denotational semantics of exceptions is defined 
as a model of this specification. 
This is extended to higher-order constructions 
in Section~\ref{subsec:expl-lambda}.

We often use the same notations for a feature  
in a signature and for its interpretation. 
So, the syntax of exceptions corresponds to the signature $\Sig_\exc$,
while the semantics of exceptions is defined as a model 
of a specification $\Ss_\expl$. But the 
signature underlying $\Ss_\expl$ is different form $\Sig_\exc$: 
this mismatch is due to the fact that 
the exceptions form a computational effect. 
The whole paper can be seen as a way to reconcile both points of view.
This can be visualized by Figure~\ref{fig:expl-syn-sem}, 
with the signature for exceptions $\Sig_\exc$ on one side
and the specification $\Ss_\expl$ with its model $M_\expl$ on the other side; 
the aim of Section~\ref{sec:deco} will be to fill the gap between
$\Sig_\exc$ and $\Ss_\expl$ by introducing 
new features in the middle, see Figure~\ref{fig:deco-syn-sem}.

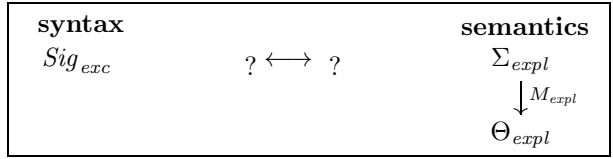
\begin{figure}[!ht]   
$$ \begin{array}{|c|} 
\hline 
  \xymatrix@C=3pc@R=0pc{
  \txt{ \textbf{syntax} } & & \txt{ \textbf{semantics} } \\
  \Sig_\exc & \txt{ ? }\longleftrightarrow\txt{ ? } & 
    \Ss_\expl \ar[ddd]^{M_\expl} \\ \\ \\
  & & \Tt_{\expl}  \\  
  }  \\
\hline 
\end{array}$$
\caption{Syntax and semantics of exceptions} 
\label{fig:expl-syn-sem} 
\end{figure}

\subsection{Signature for exceptions} 
\label{subsec:expl-sig}

The syntax for exceptions in computer languages depends on the language:  
the keywords for raising exceptions may be either 
\texttt{raise} or \texttt{throw}, 
and for handling exceptions they may be either 
\texttt{handle}, \texttt{try-with} or \texttt{try-catch}, for instance. 
In this paper we rather use \texttt{throw} and \texttt{try-catch}. 
More precisely, 
the syntax of our language may be described in two parts: 
a \emph{pure} part and an \emph{exceptional} part. 

The pure part is a signature $\Sig_\pu$.  
The interpretation of the pure operations should 
neither raise nor handle exceptions. 
For simplicity we assume that the pure operations 
are either constants or unary; 
general $n$-ary operations will be mentioned in Section~\ref{sec:conc}. 

The signature $\Sig_\exc$ for exceptions is 
made of $\Sig_\pu$ together with the types and operations 
for raising and handling exceptions. 
In order to deal with several types of exceptions 
which can be parameterized, 
we introduce a set of indices $I$ and for each index $i\in I$ 
we choose a pure type $P_i$ called the \emph{type of parameters} 
for the exceptions of index $i$. 
The new operations in $\Sig_\exc$ are the operations 
for raising and handling operations, as follows. 

\begin{definition}
\label{defi:expl-sig}
Let $\Sig_\pu$ be a signature.  
Given a set of indices $I$ and a type $P_i$ of $\Sig_\pu$ 
for each $i\in I$, 
the \emph{signature for exceptions} $\Sig_\exc$ is made of $\Sig_\pu$
together with,  for each $i\in I$:   
a \emph{raising} (or \emph{throwing} ) operation for each type $Y$ 
in $\Sig_\pu$:
  $$ \throw{i}{Y} :P_i\to Y \;,$$ 
and a \emph{handling} (or \emph{try-catch}) operation 
for each $\Sig_\exc$-term $f:X\to Y$, 
each non-empty list of indices $(i_1,\dots,i_n)$ in $I$ 
and each family of $\Sig_\exc$-terms $g_1:P_{i_1}\to Y$, \dots, $g_n:P_{i_n}\to Y$: 
  $$ \try{f}{\catchn{i_1}{g_1}{i_n}{g_n}} : X \to Y \;.$$ 
\end{definition}

\begin{remark}
\label{rem:meaning}
The precise meaning of these operations is defined 
in Section~\ref{subsec:expl-model}. 
Roughly speaking, 
relying for instance on Java see appendix~\ref{app:java}, 
raising an exception signals an error, 
which may be ``catched'' by an exception handler, 
so that the evaluation may go on along another path. 
For raising an exception, 
$\throw{i}{Y}$ turns some parameter of type $P_i$
into an exception of index $i$, in such a way that this exception 
is considered as being of type $Y$. 
For handling an exception, 
the evaluation of $\try{f}{\catch{i}{g}} $ 
begins with the evaluation of $f$; 
if the result is not an exception then it is returned;  
if the result is an exception of index $i$ 
then this exception is catched, which means that its parameter is recovered  
and $g$ is applied to this parameter; 
otherwise the exception is returned, 
which usually produces an error message like ``uncaught exception\dots''.  
The evaluation of $ \try{f}{\catchn{i_1}{g_1}{i_n}{g_n}} $ for any $n>1$ 
is similar; it is checked whether the exception returned by $f$
has index $i_1$ or $i_2$ \dots or $i_n$ in this order,
so that whenever $i_j=i_k$ with $j<k$ the clause 
$i_k\To g_{i_k}$ is never executed. 
\end{remark}

\subsection{Ordinary values and exceptional values} 
\label{subsec:expl-sem}

In order to express the denotational semantics of exceptions, 
a major point is the distinction between two kinds of values:
the ordinary (or non-exceptional) values and the exceptions.
It follows that the operations may be classified 
according to the way they may, or may not, interchange 
these two kinds of values: 
an ordinary value may be \emph{tagged} for constructing an exception, 
and later on the tag may be cleared in order to recover the value; 
then we say that the exception gets \emph{untagged}. 
Let us introduce a set $E$ called the \emph{set of exceptions}.
For each set $X$ we consider the disjoint union $X+E$.
The denotational semantics of exceptions relies on the following facts. 
Each type $X$ in $\Sig_\exc$ is interpreted as a set~$X$.
Each term $f:X\to Y$ is interpreted as a function $f:X \to Y+E$, 
and whenever $f$ is pure this function has its image in $Y$. 
The fact that a term $f:X\to Y$ is not always interpreted 
as a function $f:X \to Y$ implies that the exceptions 
form a \emph{computational effect}. 

\begin{definition}
\label{defi:expl-fcts}
For each set $X$, an element of $X+E$ is 
an \emph{ordinary value} if it is in $X$ 
and an \emph{exceptional value} if it is in $E$.
A function $f:X \to Y+E$ or $f:X+E \to Y+E$ 
\emph{raises an exception} 
if there is some $x\in X$ such that $f(x)\in E$
and $f$ \emph{recovers from an exception} 
if there is some $e\in E$ such that $f(e)\in Y$. 
A function $f:X+E \to Y+E$ \emph{propagates exceptions} 
if $f(e)=e$ for every $e\in E$.
\end{definition}

\begin{remark}   
Clearly, a function $f:X+E \to Y+E$ which propagates exceptions
may raise an exception but cannot recover from an exception. 
Such a function $f$ is characterized by its restriction $f|_X:X \to Y+E$. 
In addition, every function $f_0:X\to Y$ can be extended 
in a unique way as a function $f:X+E \to Y+E$ which propagates exceptions;  
then $f|_X$ is the composition of $f_0$ with the inclusion of $Y$ in $Y+E$.  
\end{remark}

\begin{remark}
\label{rem:ppg}  
An important feature of a language with exceptions is that 
the interpretation of every term 
is a function which propagates exceptions; 
\emph{this function may raise exceptions but it cannot 
recover from an exception.}  
Indeed, the \emph{catch} block in a \emph{try-catch} expression 
may recover from exceptions which are raised inside the \emph{try} block, 
but if an exception is raised before the \emph{try-catch} expression 
is evaluated, this exception is propagated. 
Thus, the \emph{untagging} functions that will be introduced 
in Section~\ref{subsec:expl-expl} in order to recover from exceptions 
are not the interpretation of any term of the signature $\Sig_\exc$. 
In fact, this is also the case for the \emph{tagging} functions 
that will be used for raising exceptions. 
These tagging and untagging functions are called 
the \emph{core} functions for exceptions;  
they are \emph{private} in the sense that they 
do not appear in $\Sig_\exc$, 
but they are used for defining the 
\emph{public} operations for raising and handling exceptions
which are part of $\Sig_\exc$.
\end{remark}

\subsection{Explicit logic for exceptions} 
\label{subsec:expl-expl}

Let us define a logic with a type of exceptions 
by describing its theories. 

\begin{definition}
\label{defi:expl-dialog} 
A theory of the \emph{explicit logic for exceptions} $\DiaL_\expl$ 
is a monadic equational theory (as in Example~\ref{exam:dialog-meq}) 
with a distinguished type $E$ called the \emph{type of exceptions} 
and with a cocone 
$(\inn_X: X \to X+E \lto E : \ina_X)$ for each type $X$,  
which satisfies the coproduct universal property up to congruence:
for every cocone $(f: X \to Y \lto E :k)$ 
there is a term $\cotuple{f|k}:X+E\to Y$,
unique up to equations,  
such that $\cotuple{f|k}\circ \inn_X \equiv f$ 
and $\cotuple{f|k}\circ \ina_X \equiv k$. 
\end{definition}

\begin{definition}
\label{defi:expl-set} 
Let $E$ denote a set, then $\Set_{E,\expl}$ denotes the $\DiaL_\expl$-theory 
where types, terms and equations are the sets, functions and equalities, 
where $E$ is the set of exceptions and where for each set $X$ the cocone 
$(X \to X+E \lto E)$ is the disjoint union. 
\end{definition}

\begin{remark} 
\label{rem:expl-empt} 
In addition, it can be assumed that there is an initial type $\empt$ 
(up to congruence) in each explicit theory, 
hence a unique term $\cotu_X:\empt\to X$ 
for each type $X$ such that the cocone 
$(\id_X: X \to X \lto \empt : \cotu_X)$ is a coproduct up to congruence.
\end{remark} 

\subsection{Explicit specification for exceptions} 
\label{subsec:expl-spec}

In order to express the meaning of the raising and handling operations  
we introduce new operations (called the \emph{core} operations) and equations 
in such a way that the functions for raising and handling exceptions 
are now defined in terms of the core operations. 

\begin{definition}
\label{defi:expl-spec}
Let $\Sig_\pu$ be a signature.  
Given a set of indices $I$ and 
a type $P_i$ in $\Sig_\pu$ for each $i\in I$, 
the \emph{explicit specification for exceptions} $\Ss_\expl$ 
is the $\DiaL_\expl$-specification 
made of $\Sig_\pu$ together with for each $i\in I$:
an operation $t_i:P_i\to E$ 
called the \emph{exception constructor} or the \emph{tagging} operation 
of index $i$ and    
an operation $c_i:E \to P_i+E$ 
called the \emph{exception recovery} or the \emph{untagging} function
of index $i$, 
together with the equations 
$c_i\circ t_i \equiv \inn_{P_i}$ and 
$c_i\circ t_j \equiv \ina_{P_i} \circ t_j$ for all $j\ne i$. 
Then for each $i\in I$ the raising and handling functions 
are respectively defined using these two core operations 
as follows: 
the \emph{raising} function $ \throw{i}{Y}$ 
for each type $Y$ in $\Sig_\pu$ is:
  $$ \throw{i}{Y} = \ina_Y \circ t_i :P_i\to Y+E $$ 
and the \emph{handling} function:  
  $$ \try{f}{\catchn{i_1}{g_1}{i_n}{g_n}} : X \to Y+E $$ 
  for each term $f:X\to Y+E$, 
  each non-empty list of indices $(i_1,\dots,i_n)$ 
  and each terms $g_j:P_{i_j}\to Y+E$ for $j=1,\dots,n$ 
  is defined in two steps: 	
\begin{description}
\item[(try)] the function $ \try{f}{k} : X\to Y+E $ is defined 
    for any function $k: E\to Y+E$ by: 
  $$ \try{f}{k} \;=\; 
    \Bigcotuple{\; \inn_Y \;|\; k \;} \circ f 
  $$ 
\item[(catch)] the function $\catchn{i_1}{g_1}{i_n}{g_n}: E\to Y+E$ 
is obtained by setting $p=1$ in the family of functions 
$k_p = \catchn{i_p}{g_p}{i_n}{g_n} : E\to Y+E$
(for $p=1,\dots,n+1$) which are defined recursively by: 
  $$ k_p \;=\; \begin{cases} 
    \ina_Y & \mbox{ when } p=n+1 \\
    \bigcotuple{\; g_p \;|\; k_{p+1} \;} \circ c_{i_p} & \mbox{ when } p\leq n \\
  \end{cases} $$
\end{description}
\end{definition}

\begin{remark}
When $n=1$ we get simply: 
$$ \try{f}{\catch{i}{g}} =  
  \Bigcotuple{\inn_Y |
       \bigcotuple{g| \ina_Y  
               } \circ c_i 
            } \circ f  $$
which can be illustrated as follows,
with $\try{f}k$ on the left and $k=\catch{i}{g}$ on the right:
  $$ \xymatrix@C=2pc@R=1.5pc{
  & Y \ar[d]_{\inn} \ar[rrrd]^{\inn} &&& \\
  X \ar[r]^{f} & 
    Y+E \ar[rrr]^(.4){\cotuple{\inn| k } } &
    \ar@{}[ul]|(.4){=} \ar@{}[dl]|(.4){=} && Y+E \\ 
  & E \ar[u]^{\ina} \ar[rrru]_{k} &&& \\
  } \qquad 
	\xymatrix@C=2pc@R=1.5pc{
  & P_i \ar[d]_{\inn} \ar[rrrd]^{g} &&& \\
  E \ar[r]^(.4){c_i} & 
    P_i+E 
      \ar[rrr]^(.4){\cotuple{g|\ina}} &
    \ar@{}[ul]|(.4){=} \ar@{}[dl]|(.4){=} && Y+E \\ 
  & E \ar[u]^{\ina} \ar[rrru]_{\ina} &&& \\
  }  $$
\end{remark}

\begin{remark}
About the handling function $ \try{f}{\catchn{i_1}{g_1}{i_n}{g_n}} $, 
it should be noted that each $g_i$ may itself raise exceptions
and that the indices $i_1,\dots,i_n$ form a list: 
they are given in this order 
and they need not be pairwise distinct. 
It is assumed that this list is non-empty 
because it is the usual choice in programming languages,
however it would be easy to drop this assumption. 
\end{remark}

\subsection{The intended semantics of exceptions} 
\label{subsec:expl-model}

As usual, a \emph{$\Sig$-algebra} $M$, for any signature $\Sig$, 
is made of a set $M(X)$ for each type $X$ in $\Sig$
and a function $M(f):M(X_1)\times\dots\times M(X_n) \to M(Y)$ 
for each operation $f:X_1,\dots,X_n \to Y$.  
As in Definition~\ref{defi:expl-spec}, let $\Sig_\pu$ be a signature 
and let $\Ss_\expl$ be the explicit specification for exceptions
associated to a family of pure types $(P_i)_{i\in I}$. 

\begin{definition}
\label{defi:expl-mod} 
Given a $\Sig_\pu$-algebra $M_\pu$, 
the \emph{model of exceptions} $M_\expl$ of $\Ss_\expl$ extending $M_\pu$ 
has its values in $\Set_{E,\expl}$; 
it coincides with $M_\pu$ on $\Sig_\pu$,
it interprets the type $E$ as 
the disjoint union $E=\sum_{i\in I}P_i$ 
and the tagging operations $t_i:P_i\to E$ as the inclusions.  
\end{definition}

It follows that the interpretation of the tagging operation 
maps a non-exceptional value $a\in P_i$ to an exception $t_i(a)\in E$
(for clarity we keep the notation $t_i(a)$ instead of $a$).
Then, because of the equations, 
the interpretation of the untagging operation $c_i:E \to P_i$ 
must proceed as follows: 
it checks whether its argument $e$ is in the image of $t_i$,  
if this is the case 
then it returns the parameter $a\in P_i$ such that $e=t_i(a)$, 
otherwise it propagates the exception~$e$. 
It is easy to check that the next Definition corresponds to 
the description of the mechanism of exceptions in Java: 
see remark~\ref{rem:meaning} and Appendix~\ref{app:java}.

\begin{definition}
\label{defi:expl-model}
Given a signature $\Sig_\pu$ and a $\Sig_\pu$-algebra $M_\pu$, 
the \emph{intended semantics of exceptions} is 
the model $M_\expl$ of the specification $\Ss_\expl$ extending $M_\pu$. 
\end{definition}

\begin{remark}
\label{rem:expl-model} 
Let $\Sig_\exc$ be the signature for exceptions as in 
Definition~\ref{defi:expl-sig}.
It follows from Definition~\ref{defi:expl-model} that the intended semantics 
of exceptions cannot be seen as a $\Sig_\exc$-algebra.
Indeed, although there is no type of exceptions in $\Sig_\exc$, 
the operation $\throw{i}{Y} :P_i\to Y$ in $\Sig_\exc$
has to be interpreted as a function $\throw{i}{Y} :P_i\to Y+E$, 
where the set of exceptions $E$ is usually non-empty. 
\end{remark}

\subsection{About higher-order constructions} 
\label{subsec:expl-lambda}

Definition~\ref{defi:expl-model} can easily be extended to 
a functional language. 
In order to add higher-order features to our explicit logic, 
let us introduce a functional type $Z^W$ 
for each types $W$ and $Z$. 
Then each $\varphi:W\to Z+E$ gives rise to 
$\lambda x.\varphi:\unit \to (Z+E)^W$,  
which does not raise exceptions.  
It follows that 
$\try{\lambda x.\varphi}{\catchn{i_1}{g_1}{i_n}{g_n}} \equiv \lambda x.\varphi$,
which is the intended meaning of exceptions in 
functional languages like ML \cite{sml}. 

\section{Exceptions as a computational effect}
\label{sec:deco}

According to Definition~\ref{defi:expl-model}, 
the intended semantics of exceptions can be defined in the explicit logic 
as a model $M_\expl$ of the explicit specification $\Ss_\expl$.  
However, by introducing a type of exceptions, 
the explicit logic does not take into account 
the fact that the exceptions form a computational effect: 
the model $M_\expl$ cannot be seen as an algebra of the signature $\Sig_\exc$ 
for exceptions (Definition~\ref{defi:expl-sig}) 
since (denoting $X$ for $M_\expl(X)$ for each type $X$) 
the operation $\throw{i}{Y} :P_i\to Y$
is interpreted as a function from $P_i $ to $Y+E$  
instead of from $P_i $ to $Y$: this is a fundamental remark of Moggi 
in \cite{Mo91}. 

In this Section we build another logic $\DiaL_\deco$,  
called the \emph{decorated} logic for exceptions, 
and a decorated specification $\Ss_\deco$ for exceptions
which reconciles the syntax and the semantics: 
$\Ss_\deco$ fits with the syntax since it has no type of exceptions,
and it provides the intended semantics because this semantics 
can be seen as a model $M_\deco$ of $\Ss_\deco$. 
In the decorated logic the terms and the equations are classified,
or \emph{decorated}, and their interpretation depends on their decoration. 

The decorated logic is defined in Section~\ref{subsec:deco-deco}. 
In Section~\ref{subsec:deco-spec} we define 
the decorated specification $\Ss_\deco$ and  
the model $M_\deco$ of $\Ss_\deco$ and we prove that 
$M_\deco$ provides the intended semantics of exceptions. 
Besides, we show in Section~\ref{subsec:deco-app} that it is easy 
to relate the decorated specification $\Ss_\deco$
to the signature for exceptions $\Sig_\exc$;  
for this purpose we introduce a logic $\DiaL_\app$,
called the \emph{apparent} logic,  
which is quite close to the monadic equational logic. 
This is illustrated by Figure~\ref{fig:deco-syn-sem},
which extends Figure~\ref{fig:expl-syn-sem} by filling
the gap between syntax and semantics.
This is obtained by adding two morphisms of logic, 
$F_\undeco: \DiaL_\deco \to \DiaL_\app $ on the syntax side 
and $F_\expand: \DiaL_\deco \to \DiaL_\expl $ on the semantics side.
The rules of the decorated logic are used for proving 
some properties of exceptions in Section~\ref{subsec:deco-proof}.
The decorated logic is extended to higher-order features 
in Section~\ref{subsec:deco-lambda}.

\begin{figure}[!ht]   
  $$ \begin{array}{|c|}
  \hline
  \xymatrix@C=6pc@R=0pc{
  \qquad\quad \txt{ \textbf{syntax} } & \txt{ \textbf{syntax} } \ar@{=>}[l] & \\
  & \txt{ \textbf{semantics} } \ar@{<=>}[r] & \txt{ \textbf{semantics} } \\
  \qquad\quad \DiaL_\app & 
    \DiaL_\deco \ar[l]_{F_\undeco} \ar[r]^{F_\expand} & 
    \DiaL_\expl  \\ \\
  \Sig_\exc \subseteq \Ss_\app & 
    \Ss_\deco \ar@{|->}[l]_{F_\undeco} \ar@{|->}[r]^{F_\expand} \ar[ddd]_{M_\deco} & 
    \Ss_\expl \ar[ddd]^{M_\expl} \\ \\ \\
  & \Tt_{\deco} &
  \Tt_{\expl} \ar@{|->}[l]_{G_\expand} \\  
  }  \\
  \hline
  \end{array} $$ 
\caption{Syntax and semantics of exceptions, reconciled} 
\label{fig:deco-syn-sem} 
\end{figure}
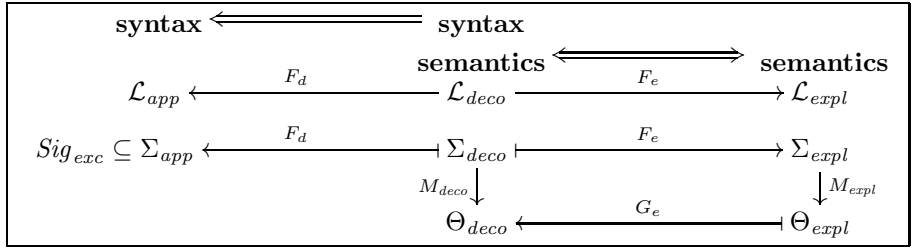

\subsection{Decorated logic for exceptions}
\label{subsec:deco-deco}

Here we define the decorated logic for exceptions $\DiaL_\deco$,
by giving its syntax and its inference rules, and we define 
a morphism from $\DiaL_\deco$ to $\DiaL_\expl$ 
for expliciting the meaning of the decorations. 
The syntax of $\DiaL_\deco$ consists in types, terms and equations,
like $\DiaL_\meq$ in Example~\ref{exam:dialog-meq}, 
but with three kinds of terms and two kinds of equations. 
The terms are decorated by $\pure$, $\ppg$ and $\ctc$ used as superscripts,
they are called respectively 
\emph{pure} terms, \emph{propagators} and \emph{catchers}. 
The equations are denoted by two distinct relational symbols, 
$\eqs$ for \emph{strong} equations and $\eqw$ for  \emph{weak} equations. 

The \emph{expansion} functor is 
the locally presentable functor $F_{\expand,S}:\catS_\deco\to\catS_\expl$ 
defined in Figure~\ref{fig:expansion} by  
mapping each elementary decorated specification 
(type, decorated term, decorated equation)
to an explicit specification. 
Note: in the explicit specifications the type of exceptions $E$ 
may be duplicated for readability, and 
the superscript $\dec$ stands for any decoration. 
Thus, the expansion provides a meaning for the decorations: 
\begin{description}
\item[$\pure$] a \emph{pure} term may neither raise exceptions 
nor recover form exceptions,
\item[$\ppg$] a \emph{propagator} may raise exceptions but is not allowed to 
recover from exceptions, 
\item[$\ctc$] a \emph{catcher}  
may raise exceptions and recover form exceptions.  
\item[$(\eqs)$] a \emph{strong} equation is an equality of functions 
both on ordinay values and on exceptions
\item[$(\eqw)$]  a \emph{weak} equation is an equality of functions 
only on ordinay values, maybe not on exceptions. 
\end{description}

\begin{remark}
\label{rem:for-short}
It happens that the image of a decorated term by the expansion morphism 
can be characterized by a term, so that we can say ``for short'' that 
the expansion of a catcher $f^\ctc:X\to Y$ ``is'' $f:X+E\to Y+E$, 
the expansion of a propagator $f^\ppg:X\to Y$ ``is'' $f_1:X\to Y+E$
where $f_1=f\circ \inn_X$,  
and the expansion of a pure term $f^\pure:X\to Y$ ``is'' $f_0:X\to Y$.
In a similar way, we say that the expansion of a type $Z$ ``is'' $Z$.
This is stated in the last column of Figure~\ref{fig:expansion}. 
However this may lead to some misunderstanding. Indeed, 
while the image of a specification by the expansion morphism 
must be a specification, the image of a type does not have to be a type
and the image of a term does not have to be a term. 
\end{remark}

\begin{figure}[!ht]   
\renewcommand{\arraystretch}{1.5}   
$$ \begin{array}{|l|l|l|l|} 
\hline 
& \Ss_\deco & F_{\expand,S} \Ss_\deco & F_{\expand,S} \Ss_\deco  \mbox{``for short''} \\ 
\hline 
\mbox{type} & 
\xymatrix@R=.5pc{
\\ 
Z  \\
} & 
\xymatrix@R=.5pc{
Z \ar[d]^{\inn}  \\
Z+E   \\
E \ar[u]_{\ina}  \\
}  &
\xymatrix@R=.5pc{
\\ 
Z  \\
} \\
\hline 
\mbox{catcher} & 
\xymatrix@R=.5pc{
\\ 
X \ar[r]^{f^\ctc} & Y \\
} & 
\xymatrix@R=.5pc{
X \ar[d] & Y \ar[d] \\
X+E \ar[r]^{f} & Y+E  \\
E \ar[u] & E \ar[u] \\
} &
\xymatrix@R=.5pc@C=3pc{
\\ 
X+E \ar[r]^{f} & Y+E \\
}   \\
\hline 
\mbox{propagator} & 
\xymatrix@R=.5pc{
\\ 
X \ar[r]^{f^\ppg} & Y \\
} & 
\xymatrix@R=.5pc{
X \ar[d] & Y \ar[d] \\
X+E \ar[r]^{f} & Y+E  \\
E \ar[u] \ar[r]_{\id} \ar@{}[ru]|{\equiv} & E \ar[u] \\
} &
\xymatrix@R=.5pc@C=5pc{
\\ 
X \ar[r]^{f_1=f\;\circ\;\inn} & Y+E \\
}   \\
\hline 
\mbox{pure term} & 
\xymatrix@R=.5pc{
\\ 
X \ar[r]^{f^\pure} & Y \\
} & 
\xymatrix@R=.5pc{
X \ar[d] \ar[r]^{f_0} \ar@{}[rd]|{\equiv} & Y \ar[d] \\
X+E \ar[r]^(.6){f} & Y+E  \\
E \ar[u] \ar[r]_{\id} \ar@{}[ru]|{\equiv} & E \ar[u] \\
} &
\xymatrix@R=.5pc@C=5pc{
\\ 
X \ar[r]^{f_0} & Y \\
}   \\
\hline 
\mbox{strong equation} & 
\begin{array}{l} f^\dec \eqs g^\dec : \\
  \qquad X \to Y \end{array} &
\begin{array}{l} f\equiv g: \\ 
  \qquad X+E \to Y+E \end{array} & 
f\equiv g \\
\hline 
\mbox{weak equation} & 
\begin{array}{l} f^\dec \eqw g^\dec : \\
  \qquad X \to Y \end{array} &
\begin{array}{l} f\circ\inn_X \equiv g\circ\inn_X: \\ 
  \qquad X \to Y+E \end{array} & 
f_1\equiv g_1 \\
\hline 
\end{array}$$
\renewcommand{\arraystretch}{1}
\caption{The expansion morphism} 
\label{fig:expansion} 
\end{figure}

\begin{figure}[!ht]   
\renewcommand{\arraystretch}{2.3}   
$$ \begin{array}{|c|} 
\hline 
\multicolumn{1}{|l|}
  {\text{(a) Monadic equational rules for exceptions (first part)} } \\ 
\dfrac{f:X\to Y \quad g:Y\to Z}{g\circ f:X\to Z} \qquad 
\dfrac{X}{\id_X:X\to X } \\ 
\dfrac{f:X\to Y \quad g:Y\to Z \quad h:Z\to W}
  {h\circ (g\circ f) \eqs (h\circ g)\circ f} \qquad 
\dfrac{f:X\to Y}{f\circ \id_X \eqs f} \qquad 
\dfrac{f:X\to Y}{\id_Y\circ f \eqs f} \\ 
\dfrac{f}{f \eqs f} \qquad 
\dfrac{f \eqs g}{g \eqs f} \qquad 
\dfrac{f \eqs g \quad g \eqs h}{f \eqs h} \\ 
\dfrac{f:X\to Y \quad g_1\eqs g_2:Y\to Z}
  {g_1\circ f \eqs g_2\circ f :X\to Z}  \qquad 
\dfrac{f_1\eqs f_2:X\to Y \quad g:Y\to Z}
  {g\circ f_1 \eqs g\circ f_2 :X\to Z} \\ 
\hline 
\multicolumn{1}{|l|}
  {\text{(b) Monadic equational rules for exceptions (second part)} } \\ 
\dfrac{f^\pure}{f^\ppg} \qquad 
\dfrac{f^\ppg}{f^\ctc} \qquad 
\dfrac{X}{\id_X^\pure } \qquad 
\dfrac{f^\pure \quad g^\pure}{(g\circ f)^\pure}  \qquad 
\dfrac{f^\ppg \quad g^\ppg}{(g\circ f)^\ppg} \\  
\dfrac{f^\ppg \eqw g^\ppg}{f \eqs g} \qquad
\dfrac{f \eqs g}{f \eqw g} \qquad  
\dfrac{f}{f \eqw f} \qquad 
\dfrac{f \eqw g}{g \eqw f} \qquad 
\dfrac{f \eqw g \quad g \eqw h}{f \eqw h} \\ 
\dfrac{f^\pure:X\to Y \quad g_1\eqw g_2:Y\to Z}
  {g_1\circ f \eqw g_2\circ f }  \qquad 
\dfrac{f_1\eqw f_2:X\to Y \quad g:Y\to Z}
  {g\circ f_1 \eqw g\circ f_2 } \\
\hline 
\multicolumn{1}{|l|}
  {\text{(c) Rules for the propagation of exceptions} } \\ 
\dfrac{k^\ctc:X\to Y}{\toppg{k}^\ppg:X\to Y} \qquad  
\dfrac{k^\ctc:X\to Y}{\toppg{k} \eqw k} \\
\hline 
\multicolumn{1}{|l|}
  {\text{(d) Rules for a decorated initial type $\empt$} } \\ 
\dfrac{X}{\cotu_X:\empt \to X} \qquad  
\dfrac{X}{\cotu_X^\pure} \qquad 
\dfrac{f:\empt \to Y}{f \eqw \cotu_Y} \\
\hline 
\multicolumn{1}{|l|}
  {\text{(e) Rules for case distinction with respect to $X+\empt$}  } \\ 
\;\dfrac{g^\ppg\!:\!X\!\to\! Y \;\; k^\ctc\!:\!\empt\!\to\! Y}
  {\cotuple{g\,|\,k}^\ctc\!:\!X\to Y}
\;\; 
\dfrac{g^\ppg\!:\!X\!\to\! Y \;\; k^\ctc\!:\!\empt\!\to\! Y}
  {\cotuple{g\,|\,k} \eqw g}
\;\;
\dfrac{g^\ppg\!:\!X\!\to\! Y \;\; k^\ctc\!:\!\empt\!\to\! Y}
  {\cotuple{g\,|\,k} \circ \cotu_X \eqs k } \; \\ 
\dfrac{g^\ppg : X\to Y \quad k^\ctc : \empt\to Y 
  \quad f^\ctc : X\to Y \quad  f \eqw g \quad 
  f \circ \cotu_X \eqs k}{f\eqs \cotuple{g\,|\,k}} \\
\hline 
\multicolumn{1}{|l|}
  {\text{(f) Rules for a constitutive coproduct $(q_i^\ppg:X_i\to X)_i$}}\\
\dfrac {(f_i^\ppg : X_i\to Y)_i} 
  {\cotuple{f_i}_i^\ctc : X\to Y} \qquad   
\dfrac {(f_i^\ppg : X_i\to Y)_i} 
  {\cotuple{f_j}_j \circ q_i \eqw f_i } \\ 
\dfrac{(f_i^\ppg : X_i\to Y)_i \quad f^\ctc : X\to Y \quad 
  \forall i \; f \circ q_i \eqw f_i}
  {f \eqs \cotuple{f_j}_j }  \\
\hline 
\end{array}$$
\renewcommand{\arraystretch}{1}
\caption{Decorated rules for exceptions} 
\label{fig:rules} 
\end{figure}

The rules of $\DiaL_\deco$ are given in Figure~\ref{fig:rules}. 
The decoration properties are often grouped 
with other properties: for instance, ``$f^\ppg \eqw g^\ppg$''
means ``$f^\ppg$ and $g^\ppg$ and $f \eqw g$''; 
in addition, the decoration $\ctc$ is usually dropped, since the rules
assert that every term can be seen as a catcher.  
According to Definition~\ref{defi:dialog-morphism}, 
the expansion morphism maps each inference rule of $\DiaL_\expl$ 
to a proof in $\DiaL_\expl$; this provides the meaning of the decorated rules: 
\begin{description}
\item{(a)} 
The first part of the decorated monadic equational rules for exceptions 
are the rules for the monadic equational logic; 
this means that the catchers satisfy the monadic equational rules
with respect to the strong equations. 
\item{(b)} 
The second part of the decorated monadic equational rules for exceptions 
deal with the conversions between decorations 
and with the equational-like properties of pure operations, propagators and 
weak equations. Every strong equation is a weak one 
while every weak equation between propagators is a strong one. 
Weak equations do not form a congruence since 
the substitution rule holds only when the substituted term is pure. 
\item{(c)} 
The rules for the propagation of exceptions  
build a propagator $\toppg{k}$ from any catcher $k$. 
The expansion of $\toppg{k}$ is defined as 
$\cotuple{k\circ \inn_X | \ina_X}:X+E\to Y+E$: 
it coincides with the expansion of $k$ on $X$ and it propagates
exceptions without catching them, otherwise. 
\item{(d)} 
The rules for a decorated initial type $\empt$ 
together with the rules in (b) imply that every propagator from $\empt$ 
to any $X$ is strongly equivalent to $\cotu_X$. 
The expansion of $\empt$ and $\cotu_X^\pure$ 
are the initial type $\empt$ and the term $\cotu_X$, respectively, 
as in remark~\ref{rem:expl-empt}. 
\item{(e)} 
The pure coproduct ($\id_X:X\to X+\empt \lto \empt : \cotu_X$)  
has decorated coproduct properties which are given by 
the rules for the case distinction with respect to $X+\empt$.  
The expansion of $\cotuple{g|k}^\ctc : X\to Y$ 
is the case distinction $\cotuple{g_1|k}: X+E\to Y+E$ with respect to $X+E$ 
(where $\empt+E$ is identified with $E$, so that $k:E\to Y+E$).
This can be illustrated as follows, 
by a diagram in the decorated logic (on the left)
or in the explicit logic (on the right); 
more details are given in Remark~\ref{rem:expansion}. 
\begin{equation}
\label{diag:expansion} 
  \xymatrix@C=2.5pc@R=1.5pc{
  X \ar[d]_(.4){\id^\pure} \ar[rrrd]^{g^\ppg} &&& \\
  X \ar[rrr]^(.4){\cotuple{g|k}^\ctc } &
    \ar@{}[ul]|(.4){\eqw} \ar@{}[dl]|(.4){\eqs} && Y \\ 
  \empt \ar[u]^{\cotu^\pure} \ar[rrru]_{k^\ctc} &&& \\
  } \qquad\qquad \xymatrix@C=2pc@R=1.5pc{
  X \ar[d]_{\inn} \ar[rrrd]^{g_1} &&& \\
  X+E \ar[rrr]^(.4){\cotuple{g_1| k } } &
    \ar@{}[ul]|(.4){\equiv} \ar@{}[dl]|(.4){\equiv} && Y+E \\ 
  E \ar[u]^{\ina} \ar[rrru]_{k} &&& \\
  } 
\end{equation} 
\item{(f)} 
The rules for a constitutive coproduct 
build a catcher from a family of propagators. 
Whenever $(q_i^\ppg:X_i\to X)_i$ is a constitutive coproduct 
the family $(q_{i,1}:X_i\to X+E)_i$ is a coproduct with respect to 
the explicit logic. 
\end{description}

\begin{remark}
\label{rem:expansion}
Let us give some additional information on the expansion of 
the decorated rules (e) in Figure~\ref{fig:rules},
i.e., the decorated rules for case distinction with respect to $X+\empt$.
According to the definition of the expansion morphism 
on specifications (Figure~\ref{fig:expansion}) 
since the cocone 
$(\id_X^\pure: X \to X+\empt \lto \empt : \cotu_X^\pure)$
is made of pure terms, we can say ``for short'' that 
its expansion ``is'' simply
$(\id_{X,0}: X \to X+\empt \lto \empt : \cotu_{X,0})$. 
However in order to check that  
the decorated rules (e) in Figure~\ref{fig:rules}
are mapped by the expansion morphism to explicit proofs 
we have to take into account another coproduct in the explicit logic. 
Rules (e) state that
\textit{for each propagator $g^\ppg:X\to Y$ and each catcher $k^\ctc:\empt\to Y$ 
there is a catcher $h^\ctc: X\to Y$ ($h$ is denoted $\cotuple{g|k} $
in Figure~\ref{fig:rules}) 
such that $h \eqw g$ and $h \circ \cotu_X \eqs k$,
and that in addition $h$ is,  up to strong equivalence, 
the unique catcher satisfying these conditions.}  
Thus, according to Figure~\ref{fig:expansion},  
the expansion of these rules must be such that 
\textit{for each terms $g_1:X\to Y+E$ and $k:E\to Y+E$ 
there is a term $h: X+E\to Y+E$ 
such that $h\circ \inn_X \equiv g\circ \inn_X$ and 
$h \circ \ina_X \equiv k$, 
and that in addition $h$ is,  up to equivalence, 
the unique term satisfying these conditions.}
Clearly, this is satisfied when $h=\cotuple{g_1|h}$ 
is obtained by case distinction with respect to the coproduct 
$(\inn_X: X \to X+E \lto E : \ina_X)$. 
It follows that we can also say, ``for short'', that the image 
of the coproduct 
$(\id_X: X \to X+\empt \lto \empt : \cotu_X)$
by the expansion morphism ``is'' the coproduct 
$(\inn_X: X \to X+E \lto E : \ina_X)$, as in diagram~(\ref{diag:expansion}). 
\end{remark}

The decorated rules are now used for proving a lemma
that will be used in Section~\ref{subsec:deco-spec}. 

\begin{lemma}
\label{lem:coprod-cotu} 
For each propagator $g^\ppg:X\to Y$ we have 
$g\circ \cotu_X \eqs \cotu_Y$ and $ g \eqs \cotuple{g\,|\,\cotu_Y}$. 
\end{lemma}

\begin{proof}
In these proofs the labels refer to the kind of rules which are used:
either $(a)$, $(b)$, $(d)$ or $(e)$.  
First, let us prove that $g\circ \cotu_X \eqs \cotu_Y$: 
\small
\begin{prooftree}
\AxiomC{$X$}
\LeftLabel{$(d)\;$} 
\UnaryInfC{$\cotu_X : \empt \to X$}
    \AxiomC{$g:X\to Y$}
  \LeftLabel{$(a)\;$} 
  \BinaryInfC{$g\circ \cotu_X : \empt \to Y$}
  \LeftLabel{$(d)\;$} 
  \UnaryInfC{$g\circ \cotu_X \eqw \cotu_Y$}
         \AxiomC{$g^\ppg$}    
                      \AxiomC{$X$}    
                      \LeftLabel{$(d)\;$} 
                      \UnaryInfC{$\cotu_X^\pure$}
                      \LeftLabel{$(b)\;$} 
                      \UnaryInfC{$\cotu_X^\ppg$}
                \LeftLabel{$(b)\;$} 
                \BinaryInfC{$(g\circ \cotu_X)^\ppg$}
                            \AxiomC{$Y$}    
                            \LeftLabel{$(d)\;$} 
                            \UnaryInfC{$\cotu_Y^\pure$}
                            \LeftLabel{$(b)\;$} 
                            \UnaryInfC{$\cotu_Y^\ppg$}
      \LeftLabel{$(b)\;$} 
      \TrinaryInfC{$g\circ \cotu_X \eqs \cotu_Y$}
\end{prooftree}
\normalsize
This first result is the unique non-obvious part in the proof 
of $g\eqs\cotuple{g\,|\,\cotu_Y}$: 
\small
\begin{prooftree}
\AxiomC{$g^\ppg : X\to Y$} 
\AxiomC{$Y$}    
\LeftLabel{$(d)\;$} 
\UnaryInfC{$\cotu_Y^\pure : \empt\to Y$}
\LeftLabel{$(b)\;$} 
\UnaryInfC{$\cotu_Y^\ppg : \empt\to Y$}
\LeftLabel{$(b)\;$} 
\UnaryInfC{$\cotu_Y^\ctc : \empt\to Y$}
\AxiomC{$g^\ppg : X\to Y$}
\LeftLabel{$(b)\;$} 
\UnaryInfC{$g^\ctc : X\to Y$}
\AxiomC{$g$}
\LeftLabel{$(b)\;$} 
\UnaryInfC{$g \eqw g$}
\AxiomC{$\vdots$}
\UnaryInfC{$g\circ \cotu_X \eqs \cotu_Y$}
\LeftLabel{$(e)\;$} 
\QuinaryInfC{$g\eqs \cotuple{g\,|\,\cotu_Y}$} 
\end{prooftree}
\normalsize
\end{proof}

\begin{remark}
\label{rem:sketches}
The morphism of limit sketches $\ske:\skE_S\to\skE_T$
which induces the decorated logic is easily guessed.
This is outlined below, more details are given in  
a similar exercice in \cite{DD12-param}. 
The description of $\skE_S$ can be read from 
the second column of Figure~\ref{fig:expansion}. 
There is in the limit sketch $\skE_S$ a point for each 
elementary decorated specification 
and an arrow for each morphism between the elementary specifications, 
in a contravariant way. 
For instance $\skE_S$ has points \emph{type} and \emph{catcher}, 
and it has arrows \emph{source} and \emph{target}
from \emph{catcher} to \emph{type}, 
corresponding to the morphisms from 
the decorated specification $Z$ to the 
decorated specification $f^\ctc:X\to Y$ 
which map $Z$ respectively to $X$ and $Y$.  
As usual, some additional points, arrows and distinguished cones  
are required in $\skE_S$.  
The description of $\ske$ can be read from Figure~\ref{fig:rules}. 
The morphism $\ske$ adds inverses to arrows in $\skE_S$ 
corresponding to the inference rules, 
in a way similar to Example~\ref{exam:dialog-meq} 
but in a contravariant way. 
\end{remark}

\begin{remark}
\label{rem:dual}
In  the short note \cite{DDFR12a-short} it is checked that,
from a denotational point of view, 
the functions for tagging and untagging exceptions are respectively 
\emph{dual},
in the categorical sense, 
to the functions for looking up and updating states.  
This duality relies on the fact that 
the states are \emph{observed} thanks to the lookup operations 
while dually the exceptions are \emph{constructed} thanks to 
the tagging operations.
Thus, the duality between states and exceptions stems from the duality 
between the comonad $X\times S$ (for some fixed $S$)
and the monad $X+E$ (for some fixed $E$). 
It happens that this duality also holds from the decorated 
point of view. 

Most of the decorated rules for exceptions
are dual to the decorated rules for states in \cite{DDFR12b-state}. 
For instance, the unique difference between 
the monadic equational rules for exceptions 
(parts (a) and (b) of Figure~\ref{fig:rules}) 
and the dual rules for states in \cite{DDFR12b-state} 
lies in the congruence rules for the weak equations:
for states the replacement rule is restricted to pure $g$,
while for exceptions it is 
the substitution rule which is restricted to pure $f$. 
The rules for a decorated initial type 
and for a constitutive coproduct
(parts (d) and (f) of Figure~\ref{fig:rules}) 
are respectively dual to the rules for a decorated final type 
and the rules for an observational product in \cite{DDFR12b-state}. 
The rules for the propagation of exceptions 
and for the case distinction with respect to $X+\empt$ 
(parts (c) and (e) of Figure~\ref{fig:rules}) 
are used only for the construction 
of the handling operations from the untagging operations; 
these rules have no dual in \cite{DDFR12b-state} for states. 
\end{remark}

\begin{remark}
\label{rem:monad} 
For a while, let us forget about the three last families of rules
in  Figure~\ref{fig:rules}, 
which involve some kind of decorated coproduct. 
Then any monad $T$ on any category $\catC$ provides 
a decorated theory $\catC_T$, as follows. 
The types are the objects of $\catC$,
a pure term $f^\pure:X\to Y$ is a morphism  $f:X\to Y$ in $\catC$,
a propagator $f^\ppg:X\to Y$ is a morphism  $f:X\to TY$ in $\catC$,
a catcher $f^\ctc:X\to Y$ is a morphism  $f:TX\to TY$ in $\catC$.
The conversion from pure to propagator uses the unit of $T$ 
and the conversion from propagator to catcher uses the multiplication of $T$. 
Composition of propagators is done in the Kleisli way. 
A strong equation $f^\ctc\eqs g^\ctc:X\to Y$ is an equality 
$f=g:TX\to TY$ in $\catC$  
and a weak equation $f^\ctc\eqw g^\ctc:X\to Y$ is an equality 
$f\circ\eta_X=g\circ\eta_X:X\to TY$ in $\catC$, 
where $\eta$ is the unit of the monad. 
It is easy to check that the decorated monadic equational rules of $\DiaL_\deco$ 
are satisfied, as well as the rules for the propagation of exceptions 
if $\toppg{k}=k\circ \eta_X:X\to TY$ for each $k:TX\to TY$. 
\end{remark}

\subsection{Decorated specification for exceptions}
\label{subsec:deco-spec}

Let us define a decorated specification $\Ss_\deco$ for exceptions,  
which (like $\Ss_\expl$ in Section~\ref{subsec:expl-spec}) 
defines the raising and handling operations 
in terms of the core tagging and untagging operations. 

\begin{definition}
\label{defi:deco-spec}
Let $\Sig_\pu$ be a signature.  
Given a set of indices $I$ and 
a type $P_i$ in $\Sig_\pu$ for each $i\in I$, 
the \emph{decorated specification for exceptions} $\Ss_\deco$ 
is the $\DiaL_\deco$-specification 
made of $\Sig_\pu$ with its operations decorated as pure 
together with, for each $i\in I$, 
a propagator $t_i^{\ppg}:P_i\to\empt$ and  
a catcher $c_i^{\ctc}: \empt \to P_i$  
with the weak equations   
$ c_i \circ t_i \eqw \id : P_i \to P_i $
and 
$ c_i \circ t_j \eqw \cotu \circ t_j : P_j \to P_i $ for all $j\ne i$.
Then for each $i\in I$
the \emph{raising} propagator 
$ (\throw{i}{Y})^\ppg :P_i\to Y $ 
for each type $Y$ in $\Sig_\pu$ 
  is:
	$$ \throw{i}{Y} = \cotu_Y \circ t_i$$ 
and the \emph{handling} propagator 
  $ (\try{f}{\catchn{i_1}{g_1}{i_n}{g_n}})^\ppg  : X \to Y $
  for each propagator $f^\ppg:X\to Y$, 
  each non-empty list of indices $(i_1,\dots,i_n)$ 
  and each propagators $g_j^\ppg:P_{i_j}\to Y$ for $j=1,\dots,n$ 
  is defined  as:
	$$ \try{f}{\catchn{i_1}{g_1}{i_n}{g_n}} = 
         \toppg{ \TRY{f}{\catchn{i_1}{g_1}{i_n}{g_n}} } $$
  from a catcher $ \TRY{f}{\catchn{i_1}{g_1}{i_n}{g_n}}:X\to Y$ 
  which is defined as follows in two steps: 	
\begin{description}
\item[(try)] the catcher $ \TRY{f}{k} : X\to Y $ is defined 
for any catcher $k: \empt\to Y$ by: 
  $$
  (\TRY{f}{k})^\ctc \;=\; 
    \cotuple{\; \id_Y^\pure \;|\; k^\ctc \;}^\ctc \circ f^\ppg 
  $$
\item[(catch)] the catcher  $\catchn{i_1}{g_1}{i_n}{g_n}: \empt\to Y$ 
is obtained by setting $p=1$ in the family of catchers 
$k_p = \catchn{i_p}{g_p}{i_n}{g_n} : \empt\to Y$ 
(for $p=1,\dots,n+1$) which are defined recursively by: 
  $$ k_p^\ctc \;=\;  
    \begin{cases} 
      \cotu_Y^\pure & \mbox{ when } p=n+1 \\
      \cotuple{\; g_p^\ppg \;|\; k_{p+1}^\ctc \;}^\ctc  \circ c_{i_p}^\ctc & 
		   \mbox{ when } p\leq n \\
   \end{cases} $$
\end{description}
\end{definition}

\begin{remark}
\label{rem:try}
Let $ h = \try{f}{\catchn{i_1}{g_1}{i_n}{g_n}}$
and $H = \TRY{f}{\catchn{i_1}{g_1}{i_n}{g_n}}$. 
Then $h$ is a propagator and $H$ is a catcher, 
and the definition of $h$ is given in terms of $H$, as $h = \toppg{H}$. 
The expansions of $h$ and $H$ are functions 
from $X+E$ to $Y+E$ which coincide on $X$
but differ on $E$: while $h$ propagates exceptions, 
$H$ catches exceptions according to the pattern 
$\catchn{i_1}{g_1}{i_n}{g_n}$. 
\end{remark}

\begin{remark}
Since $k_{n+1} = \cotu_Y$, by Lemma~\ref{lem:coprod-cotu} 
we have $\cotuple{g_n|k_{n+1}} \eqs g_n$. 
It follows that when $n=1$ and 2 we get respectively:
\begin{gather}
\label{eq:deco-handle-one}
\try{f}{\catch{i}{g}} \;\eqs\;  
   \bigtoppg{ 
    \left(\; \bigcotuple{ \id_Y \;|\; g \circ c_i } 
      \circ f \;\right) } \\
\label{eq:deco-handle-two}
\try{f}{\catchij{i}{g}{j}{h}} \;\eqs\;  
   \bigtoppg{
    \left(\;  
    \bigcotuple{\id \;|\; 
      \cotuple{g \;|\; 
        h \circ c_j} 
      \circ c_i} 
    \circ f\;\right) } 
\end{gather}
When $n=1$ this can be illustrated as follows,
with $\TRY{f}k$ on the left and $k=\catch{i}{g}$ on the right:
$$ \xymatrix@C=2.5pc@R=1.5pc{
  & Y \ar[d]_(.4){\id^\pure} \ar[rrrd]^{\id^\pure} &&& \\
  X \ar[r]^(.4){f^\ppg} & 
    Y \ar[rrr]^(.4){\cotuple{\id|k}^\ctc } &
    \ar@{}[ul]|(.4){\eqw} \ar@{}[dl]|(.4){\eqs} && Y \\ 
  & \empt \ar[u]^{\cotu^\pure} \ar[rrru]_{k^\ctc} &&& \\
  } \qquad
 \xymatrix@C=2.5pc@R=1.5pc{
  & P_i \ar[d]_(.4){\id^\pure} \ar[rrrd]^{g^\ppg} &&& \\
  E \ar[r]^(.4){c_i^\ctc} & 
    P_i  
      \ar[rrr]^(.4){\cotuple{g|\cotu}^\ctc} &
    \ar@{}[ul]|(.4){\eqw} \ar@{}[dl]|(.4){\eqs} && Y \\ 
  & \empt \ar[u]^{\cotu^\pure} \ar[rrru]_{\cotu^\pu} &&& \\
  }  $$
\end{remark}

\begin{lemma}
\label{lem:Ss} 
Let $\Sig_\pu$ be a signature, $I$ a set and  
$P_i$ a type in $\Sig_\pu$ for each $i\in I$.
Let $\Ss_\expl$ be the corresponding explicit specification for exceptions
(Definition~\ref{defi:expl-spec}) and 
$\Ss_\deco$ the corresponding decorated specification for exceptions
(Definition~\ref{defi:deco-spec}). 
Then $\Ss_\expl = F_{\expand}\Ss_\deco$.    
\end{lemma}

\begin{proof} 
This is easy to check: in Definition~\ref{defi:expl-spec} $\Ss_\expl$
is described as a colimit of elementary specifications, 
and $F_{\expand}$, as any left adjoint functor, preserves colimits. 
\end{proof}

\begin{proposition}
\label{prop:deco-expand}  
The functor $F_{\expand,S}:\catS_\deco\to\catS_\expl$  
defined in Figure~\ref{fig:expansion} is locally presentable
and it determines a morphism of logics $F_\expand:\DiaL_\deco\to\DiaL_\expl$. 
\end{proposition}

\begin{proof} 
The fact that $F_{\expand,S}$ is locally presentable is deduced 
from its definition in Figure~\ref{fig:expansion}. 
It has been checked that $F_{\expand,S}$ maps each decorated inference rule 
to an explicit proof, thus it can be extended 
as $F_{\expand,T}:\catT_\deco\to\catT_\expl$ in such a way that the pair 
$F_\expand=(F_{\expand,S},F_{\expand,T})$ is a morphism of logics. 
\end{proof}

\begin{definition}
\label{defi:deco-expand}  
The morphism $F_\expand:\DiaL_\deco\to\DiaL_\expl$ is called 
the \emph{expansion} morphism. 
\end{definition}

\subsection{The decorated model provides the intented semantics of exceptions}
\label{subsec:deco-thm}

Following Definition~\ref{defi:expl-model}, 
the intended semantics of exceptions 
is a model with respect to the explicit logic. 
Theorem~\ref{thm:deco-model}  will prove that the intended semantics 
of exceptions 
can also be expressed as a model with respect to the decorated logic. 

\begin{definition}
\label{defi:deco-set}
For any set $E$, called the \emph{set of exceptions},  
we define a decorated theory $\Set_{E,\deco}$ as follows. 
A type is a set,
a pure term $f^\pure:X\to Y$ is a function $f:X\to Y$, 
a propapagator $f^\ppg:X\to Y$ is a function $f:X\to Y+E$, 
and a catcher $f^\ctc:X\to Y$ is a function $f:X+E\to Y+E$. 
It follows that in $\Set_{E,\deco}$ 
every pure term $f:X\to Y$ gives rise to 
a propagator $\inn_Y\circ f:X\to Y+E$ 
and that every propagator $f:X\to Y+E$
gives rise to a catcher $[f|\ina_Y]:X+E\to Y+E$.
By default, $f$ stands for $f^\ctc$. 
The equations are defined when both members are catchers, 
the other cases follow thanks to the conversions above. 
A strong equation $f\eqs g:X\to Y$ is the equality of functions 
$f=g:X+E\to Y+E$
and a weak equation $f\eqw g:X\to Y$ is the equality of functions 
$f\circ\inn_X=g\circ\inn_X:X\to Y+E$.
\end{definition}

\begin{lemma}
\label{lem:Tt} 
Let $G_{\expand,T}$ be the right adjoint to $F_{\expand,T}$. 
Then $\Set_{E,\deco} = G_{\expand,T}\Set_{E,\expl}$. 
\end{lemma}

\begin{proof} 
The morphism of limit sketches $\varphi_{\expand}$, 
corresponding to the locally presentable functor $F_{\expand,T}$, 
is deduced from Figure~\ref{fig:expansion}. 
By definition of $G_{\expand,T}$  
we have $G_{\expand,T}\Set_{E,\expl} = \Set_{E,\expl} \circ \varphi_{\expand}$. 
The lemma follows by checking that  
the definition of $\Set_{E,\deco}$ (Definition~\ref{defi:deco-set}) 
is precisely the description of $\Set_{E,\expl} \circ \varphi_{\expand}$.
\end{proof}

Our main result is the next theorem, which states that 
the decorated point of view provides exactly the semantics of exceptions 
defined as a model of the explicit specification for exceptions 
in Definition~\ref{defi:expl-model}.
Thus the decorated point of view {\em is} an alternative to
the explicit point of view, as it provides the intended semantics,
but it is also closer to the syntax since the type of exceptions is no
longer explicit.

To prove this, the key point is the existence of the expansion morphism 
from the decorated to the explicit logic. 
Within the category of diagrammatic logics, the proof is simple:
it uses the fact that the expansion morphism, like every morphism in
this category, is a left adjoint functor.

\begin{theorem}
\label{thm:deco-model} 
The model $M_\deco$ of the specification $\Ss_\deco$  
with values in the theory $\Set_{E,\deco}$ in the decorated logic 
provides the intended semantics of exceptions. 
\end{theorem}

\begin{proof} 
According to Definition~\ref{defi:expl-model}, 
the intended semantics of exceptions 
is the model $M_\expl$ of $\Ss_\expl$  
with values in $\Set_{E,\expl}$ in the explicit logic. 
In addition, $M_\deco$ is a model of $\Ss_\deco$  
with values in $\Set_{E,\deco}$ in the decorated logic. 
Furthermore, we know from Lemmas~\ref{lem:Ss} and~\ref{lem:Tt}
that $\Ss_\expl = F_{\expand}\Ss_\deco$ 
and $\Set_{E,\deco} = G_{\expand}\Set_{E,\expl}$,
where $G_{\expand}$ is right adjoint to $F_{\expand}$.    
Thus, it follows from proposition~\ref{prop:deco-morphism} that 
there is a bijection between 
$ \Mod_{\DiaL_\expl}(\Ss_\expl,\Set_{E,\expl}) $ and 
$ \Mod_{\DiaL_\deco}(\Ss_\deco,\Set_{E,\deco}) $. 
Finally, it is easy to check that $M_\deco$ corresponds to $M_\expl$ 
in this bijection. 
\end{proof} 

\subsection{The decorated syntax provides the syntax of exceptions}
\label{subsec:deco-app}

The signature $\Sig_\exc$ from Definition~\ref{defi:expl-sig} 
can easily be recovered from the decorated specification $\Ss_\deco$ 
by dropping the decorations and forgetting the equations. 
More formally, this can be stated as follows. 
Let us introduce a third logic $\DiaL_\app$, 
called the \emph{apparent} logic, by dropping all the decorations 
from the decorated logic; 
thus, the apparent logic is essentially the monadic equational logic 
with an empty type. 
The fact of dropping the decorations is a morphism of logics 
$F_\undeco:\DiaL_\deco\to\DiaL_\app$.
Therefore, we can form the apparent specification
$\Ss_\app=F_\undeco\Ss_\deco$ which contains the signature for
exceptions $\Sig_\exc$. Note that, as already mentioned in
Remark~\ref{rem:expl-model}, 
the intended semantics of exceptions cannot be seen as a set-valued
model of $\Ss_\app$.

\subsection{Some decorated proofs for exceptions}
\label{subsec:deco-proof}

According to Theorem~\ref{thm:deco-model},  
the intended semantics of exceptions can be expressed as a model 
in the decorated logic. 
Now we show that the decorated logic can also be used 
for proving properties of exceptions in a concise way. 
Indeed, as for proofs on states in \cite{DDFR12b-state}, 
we may consider two kinds of proofs on exceptions: 
the \emph{explicit} proofs involve a type of exceptions,
while the \emph{decorated} proofs do not mention any type of exceptions 
but require the specification to be decorated,
in the sense of Section~\ref{sec:deco}.
In addition, the expansion morphism, from the decorated logic 
to the explicit logic, maps each decorated proof to an explicit one. 
In this Section we give some decorated proofs for exceptions, 
using the inference rules of Section~\ref{subsec:deco-deco}. 

We know from \cite{DDFR12a-short} that the properties of 
the core tagging and untagging operations for exceptions 
are dual to the properties of 
the lookup and update operations for states.
Thus, we may reuse the decorated proofs involving states 
from \cite{DDFR12b-state}. 
Starting from any one of the seven equations for states
in \cite{PP02} we can dualize this equation 
and derive a property about raising and handling 
exceptions. This is done here for 
the \emph{annihilation catch-raise} and for 
the \emph{commutation catch-catch} properties. 


On states, the \emph{annihilation lookup-update} property  
means that updating any location with the content of this location
does not modify the state. 
A decorated proof of this property is given in \cite{DDFR12b-state}.
By duality we get the following \emph{annihilation untag-tag} property
(Lemma~\ref{lem:ci-ti}), 
which means that tagging just after untagging, 
both with respect to the same index, returns the given exception. 
Then this result is used in Proposition~\ref{prop:hi-ri} 
for proving the \emph{annihilation catch-raise} property:
catching an exception and re-raising it is like doing nothing. 

\begin{lemma}[Annihilation untag-tag]
\label{lem:ci-ti} 
For each $i\in I$:
$$ 
 t_i^\ppg \circ c_i^\ctc \eqs \id_\empt^\pure \;.
$$ 
\end{lemma}

\begin{proposition}[Annihilation catch-raise]
\label{prop:hi-ri} 
For each propagator $f^\ppg:X\to Y$ and each $i\in I$:  
$$ 
  \try{f}{\catch{i}{\throw{i}{Y}}} \eqs f   \;.
$$
\end{proposition}

\begin{proof} 
By Equation~(\ref{eq:deco-handle-one}) and 
Definition~\ref{defi:deco-spec} we have 
$ \try{f}{\catch{i}{\throw{i}{Y}}} \eqs  
   \toppg{ (\cotuple{ \id_Y | \cotu_Y \circ t_i \circ c_i } \circ f ) } $.
By Lemma~\ref{lem:ci-ti} 
$ \cotuple{ \id_Y | \cotu_Y \circ t_i \circ c_i } \eqs  
\cotuple{ \id_Y | \cotu_Y } $,
and the unicity property of $\cotuple{ \id_Y | \cotu_Y }$
implies that $\cotuple{ \id_Y | \cotu_Y } \eqs \id_Y $. 
Thus $ \try{f}{\catch{i}{\throw{i}{Y}}} \eqs \toppg{f}$.
In addition, since $\toppg{f} \eqw f$ and $f$ is a propagator 
we get $\toppg{f}\eqs f$. 
Finally, the transitivity of $\eqs$ yields the proposition. 
\end{proof} 


On states, the \emph{commutation update-update} property  
means that updating two different locations can be done in any order. 
By duality we get the following \emph{commutation untag-untag} property, 
(Lemma~\ref{lem:cj-ci}) 
which means that untagging with respect to two distinct 
exceptional types can be done in any order. 
A detailed decorated proof of the 
commutation update-update property
is given in \cite{DDFR12b-state}.
The statement of this property and its proof 
use \emph{semi-pure products}, 
which were introduced in 
\cite{DDR11-seqprod} in order to provide a decorated alternative 
to the strength of a monad. 
Dually, for the commutation untag-untag property 
we use \emph{semi-pure coproducts},
thus generalizing the rules for the coproduct $X+\empt$.

The \emph{coproduct} of two types $A$ and $B$ 
is defined as a type $A+B$ with two pure coprojections
$\copi_1^\pure:A \to A+B$ and $\copi_2^\pure:B \to A+B$, 
which satisfy the usual categorical coproduct property 
with respect to the pure morphisms. 
Then the \emph{semi-pure coproduct} of 
a propagator $f^\ppg:A\to C$
and a catcher $k^\ctc:B\to C$ is a catcher 
$\cotuple{f|k}^{\ctc}:A+B\to C$ 
which is characterized, up to strong equations, by 
the following decorated version of the coproduct property: 
$\cotuple{f|k} \circ \copi_1 \eqw f $
and $\cotuple{f|k} \circ \copi_2 \eqs k $.
Then as usual, the coproduct 
$f'+k':A+B\to C+D$ of 
a propagator $f':A\to C$
and a catcher $k':B\to D$ is the catcher 
$f'+k' = \cotuple{\copi_1\circ f\;|\; \copi_2\circ k}:A+B\to C+D$.
 
Whenever $f$ and $g$ are propagators it can be proved that 
$\toppg{\cotuple{f|g}} \eqs \cotuple{f|g}$;  
thus, up to strong equations, 
we can assume that in this case 
$\cotuple{f\;|\;g}:A+B \to C$ is a propagator; 
it is characterized, up to strong equations, by
$\cotuple{f\;|\;g} \circ \copi_1 \eqs f$ and 
$\cotuple{f\;|\;g} \circ \copi_2 \eqs g$.

\begin{lemma}[Commutation untag-untag] 
\label{lem:cj-ci} 
For each $i,j\in I$ with $i\ne j$: 
$$ 
  (c_i+ \id_{P_j})^\ctc \circ c_j^\ctc \eqs 
    (\id_{P_i} + c_j)^\ctc \circ c_i^\ctc  : \empt \to P_i+P_j
$$
\end{lemma}

\begin{proposition}[Commutation catch-catch] 
\label{prop:hj-hi}
For each $i,j\!\in\! I$ with $i\!\ne\! j$: 
$$ 
\try{f}{\catchij{i}{g}{j}{h}}\eqs \try{f}{\catchij{j}{h}{i}{g}} 
$$
\end{proposition}

\begin{proof} 
According to Equation~(\ref{eq:deco-handle-two}): 
$  \try{f}{\catchij{i}{g}{j}{h}} \eqs 
  \toppg{( 
    \cotuple{\id \;|\; 
      \cotuple{g \;|\; 
        h \circ c_j} 
      \circ c_i} 
    \circ f )} $.
Thus, the result will follow from
$ \cotuple{g \;|\; h \circ c_j} \circ c_i \eqs
  \cotuple{h \;|\;  g \circ c_i}  \circ c_j $. 
It is easy to check that 
 $ \cotuple{g \;|\; h \circ c_j}
  \eqs \cotuple{g \;|\; h} \circ (\id_{P_i} + c_j) $, 
so that
$ \cotuple{g \;|\; h \circ c_j} \circ c_i \eqs 
   \cotuple{g \;|\; h} \circ (\id_{P_i} + c_j) \circ c_i \;.$
Similarly
$ \cotuple{h \;|\; g \circ c_i} \circ c_j \eqs 
   \cotuple{h \;|\; g} \circ (\id_{P_j} + c_i) \circ c_j $
hence 
$ \cotuple{h \;|\; g \circ c_i} \circ c_j \eqs 
   \cotuple{g \;|\; h} \circ (c_i + \id_{P_j}) \circ c_j \;.$ 
Then the result follows from Lemma~\ref{lem:cj-ci}.
\end{proof} 

\subsection{About higher-order constructions} 
\label{subsec:deco-lambda}

We know from Section~\ref{subsec:expl-lambda} that 
we can add higher-order features in our explicit logic.  
This remark holds for the decorated logic as well.  
Let us introduce a functional type $Z^{W\dec}$ 
for each types $W$ and $Z$ and each decoration $\dec$ for terms.
The expansion of $Z^{W\pure}$ is $Z^W$,
the expansion of $Z^{W\ppg}$ is $(Z+E)^W$ 
and the expansion of $Z^{W\ctc}$ is $(Z+E)^{(W+E)}$. 
Then each $\varphi^\dec:W\to Z$ gives rise to 
$\lambda x.\varphi:\unit \to Z^{W\dec}$,
and a major point is that 
$\lambda x.\varphi$ is pure for every decoration $\dec$ of $\varphi$.
Informally, we can say that 
the abstraction moves the decoration from the term to the type. 
This means that the expansion of $(\lambda x.\varphi)^\pure$ 
is $\lambda x.\varphi : \unit \to F_\expand(Z^{W\dec})$,
as required: for instance 
when $\varphi^\ppg$ is a propagator the expansion of 
$(\lambda x.\varphi)^\pure$ is 
$\lambda x.\varphi : \unit \to (Z+E)^W$,
as in Section~\ref{subsec:expl-lambda}. 
Besides, it is easy to prove in the decorated logic that 
whenever $f$ is pure we get $\try{f}{\catchn{i_1}{g_1}{i_n}{g_n}} \equiv f$. 
It follows that this occurs when $f$ is a lambda abstraction: 
$\try{\lambda x.\varphi}{\catchn{i_1}{g_1}{i_n}{g_n}} \equiv \lambda
x.\varphi$, as expected in functional languages.

\section{Conclusion and future work}
\label{sec:conc}

We have presented three logics for dealing with exceptions: 
the \emph{apparent} logic $\DiaL_\app$ (Section~\ref{subsec:deco-app}) 
for dealing with the syntax, 
the \emph{explicit} logic $\DiaL_\expl$ (Section~\ref{subsec:expl-expl})
for providing the semantics of exceptions as a model  
in a transparent way, 
and the \emph{decorated} logic $\DiaL_\deco$ (Section~\ref{subsec:deco-deco}) 
for reconciling syntax and semantics. 
These logics are related by morphisms of logics   
$F_\undeco:\DiaL_\deco\to\DiaL_\app$ and $F_\expand:\DiaL_\deco\to\DiaL_\expl$. 
A similar approach can be used for other exceptions 
\cite{DD10-dialog,DDFR12b-state}. 

Future work include the following topics. 
\begin{itemize}
\item Dealing with $n$-ary operations involving exceptions. 
We can add a cartesian structure to our decorated logic
thanks to the notion of \emph{sequential product} from \cite{DDR11-seqprod}.
This notion is based on the \emph{semi-pure products},
which are dual to the semi-pure coproducts 
used in Section~\ref{subsec:deco-proof}. 
\item Adding higher-order features. 
This has been outlined in Sections~\ref{subsec:expl-lambda} 
and~\ref{subsec:deco-lambda}, 
however a more precise comparison with \cite{RT99} remains to be done. 
\item Deriving a decorated operational semantics for exceptions  
by directing the weak and strong equations. 
\item Using a proof assistant for decorated proofs.
Thanks to the morphism $F_\undeco : \DiaL_\deco \to \DiaL_\app$,  
checking a decorated proof can be split in two parts: 
first checking the undecorated proof in the apparent logic,  
second checking that the decorations can be added.
This separation simplifies the definition of the formalization
  towards a proof assistant: first formalize the syntactic rules of
  the language, second add computational effects.
\item Combining computational effects. 
Since an effect is based on a span of logics, 
the combination of effects might be based on the composition of spans. 
\end{itemize} 

\noindent \textbf{Acknowledgment.} 
We are indebted to Olivier Laurent for pointing out 
the extension of our approach to functional languages. 


\appendix

\section{Handling exceptions in Java}
\label{app:java}

Definition~\ref{defi:expl-model} relies on the following description 
of the handling of exceptions in Java \cite[Ch.~14]{java}.  
\begin{itemize}
\item[] A try statement without a finally block is executed by 
first executing the try block. 
Then there is a choice:
\begin{enumerate}
\item 
\label{tc-one}
If execution of the try block completes normally, 
then no further action is taken 
and the try statement completes normally.
\item 
\label{tc-two}
If execution of the try block completes abruptly 
because of a throw of a value $V$, 
then there is a choice:
  \begin{enumerate}
  \item 
  If the run-time type of $V$ is assignable to the parameter 
  of any catch clause of the try statement, 
  then the first (leftmost) such catch clause is selected. 
  The value $V$ is assigned to the parameter of the selected 
  catch clause, and the block of that catch clause is executed. 
    \begin{enumerate}
    \item
    \label{tc-i}
    If that block completes normally, 
    then the try statement completes normally; 
    \item 
    \label{tc-ii}
    if that block completes abruptly for any reason, 
    then the try statement completes abruptly for the same reason.
    \end{enumerate}
  \item 
  If the run-time type of $V$ is not assignable to the parameter 
  of any catch clause of the try statement, then the try statement 
  completes abruptly because of a throw of the value $V$.
  \end{enumerate}
\item 
\label{tc-three}
If execution of the try block completes abruptly for any other reason, 
then the try statement completes abruptly for the same reason.
\end{enumerate}
\end{itemize}
In fact, points~\ref{tc-i} and~\ref{tc-ii} can be merged. 
Our treatment of exceptions is similar to the one in Java 
when execution of the try block completes normally (point~\ref{tc-one}) 
or completes abruptly because of a throw of an exception of constructor 
$i\in I$ (point~\ref{tc-two}): indeed, in our framework there is no other reason
for the execution of a try block to complete abruptly (point~\ref{tc-three}). 
Thus, the description can be simplified as follows. 
\begin{itemize}
\item[] 
A try statement without a finally block is executed by 
first executing the try block. 
Then there is a choice:
\begin{enumerate}
\item 
\label{new-tc-one}
If execution of the try block completes normally, 
then no further action is taken 
and the try statement completes normally.
\item 
\label{new-tc-two}
If execution of the try block completes abruptly 
because of a throw of a value $V$, 
then there is a choice:
  \begin{enumerate}
  \item 
  If the run-time type of $V$ is assignable to the parameter 
  of any catch clause of the try statement, 
  then the first (leftmost) such catch clause is selected. 
  The value $V$ is assigned to the parameter of the selected 
  catch clause, the block of that catch clause is executed, 
	and the try statement completes in the same way as this block.
  \item 
  If the run-time type of $V$ is not assignable to the parameter 
  of any catch clause of the try statement, then the try statement 
  completes abruptly because of a throw of the value $V$.
  \end{enumerate}
\end{enumerate}
\end{itemize}
This simplified description corresponds to the definition 
of $ \try{f}{\catchn{i_1}{g_1}{i_n}{g_n}} $ in Definition~\ref{defi:expl-spec}, 
with points~\ref{new-tc-one} and~\ref{new-tc-two}  corresponding 
respectively to \textbf{(try)} and \textbf{(catch)}.

\end{document}